%% file: paper_arxiv.tex
\documentclass{article}

\usepackage{amsthm}
\usepackage{amsmath}
\usepackage{amsfonts}
\usepackage{amssymb}
\usepackage{color}
\usepackage{graphicx}
\usepackage{caption}
\usepackage{subcaption}
\usepackage{bbm}
\usepackage{dsfont}
\usepackage{float}
\usepackage[margin=1in]{geometry}

\let\chapter\section
\usepackage[ruled]{algorithm2e}
\usepackage{xcolor}
\renewcommand{\algorithmcfname}{ALGORITHM}
\usepackage[numbers,sort&compress]{natbib} 
\SetArgSty{textrm}  
\SetAlFnt{\small}
\SetAlCapFnt{\small}
\SetAlCapNameFnt{\small}
\SetAlCapHSkip{0pt}
\IncMargin{-\parindent}

\newtheorem{theorem}{Theorem}
\newtheorem{lemma}{Lemma}
\newtheorem{corollary}{Corollary}
\newtheorem{remark}{Remark}
\newtheorem{claim}{Claim}
\newtheorem{definition}{Definition}
\newtheorem{assumption}{Assumption}
\newtheorem{example}{Example}

\newif\ifdraft \drafttrue
\newif\iffull \fullfalse

\definecolor{DarkGreen}{rgb}{0.1,0.5,0.1}
\definecolor{DarkRed}{rgb}{0.5,0.1,0.1}
\definecolor{DarkBlue}{rgb}{0.1,0.1,0.5}
\definecolor{brown}{rgb}{0.5,0.3,0.1}
\usepackage[]{hyperref}
\hypersetup{
    unicode=false,          
    pdftoolbar=true,        
    pdfmenubar=true,        
    pdffitwindow=false,      
    pdftitle={},    
    pdfauthor={}
    pdfsubject={},   
    pdfnewwindow=true,      
    pdfkeywords={keywords}, 
    colorlinks=true,       
    linkcolor=DarkRed,          
    citecolor=DarkGreen,        
    filecolor=DarkRed,      
    urlcolor=DarkBlue,          
}

\newcommand{\showsection}[1]{}

\newcommand\RR{\mathbb{R}}
\newcommand\cA{\mathcal{A}}

\newcommand\cR{\mathcal{R}}

\begin{document}

\title{Efficiently characterizing games consistent with perturbed equilibrium observations}
\author{Juba ZIANI\thanks{California Institute of Technology, jziani@caltech.edu}
\and Venkat CHANDRASEKARAN\thanks{California Institute of Technology, venkatc@caltech.edu}
\and Katrina LIGETT\thanks{Hebrew University of Jerusalem and California Institute of Technology, katrina.ligett@mail.huji.ac.il}}
\date{March 22, 2017}

\maketitle

\begin{abstract}

We study the problem of characterizing the set of games that are consistent with observed equilibrium play.
Our contribution is to develop and analyze a new methodology based on convex optimization to address this problem for many classes of games and observation models of interest. Our approach provides a \textit{sharp, computationally efficient} characterization of the extent to which a particular set of observations constrains the space of games that could have generated them. This allows us to solve a number of variants of this problem as well as to quantify the power of games from particular classes (e.g., zero-sum, potential, linearly parameterized) to explain player behavior. We illustrate our approach with numerical simulations.

\end{abstract}

\input{intro}

\input{summary_results}

\input{related_work}

\input{model}

\input{convex_framework}

\input{diameter}

\input{payoff_info_recovery}

\showsection{
\input{nopayoff_degeneracy}

}

\input{extensions}

\input{simulations}

\section*{Acknowledgements}
We thank Federico Echenique, Denis Nekipelov, Matt Shum, and Vasilis Syrgkanis for extremely helpful comments and suggestions. Ziani's research was funded in part by NSF grant CNS-1254169. Chandrasekaran's research was funded in part by NSF awards CCF-1350590 and CCF-1637598, Air Force Office of Scientific Research grants FA9550-14-1-0098 and
FA9550-16-1-0210, and the Sloan research fellowship. Ligett's research was funded in part by ISF grant 1044/16, NSF grants CNS-1254169 and CNS-1518941, a subcontract under the DARPA Brandeis Project, and the Hebrew University Cybersecurity Research Center in conjunction with the Israel National Cyber Bureau in the Prime Minister's Office. Ligett's work was done in part while the author was visiting the Simons Institute for the Theory of Computing at Berkeley.

%
%

\bibliographystyle{abbrv}
\bibliography{bibliography}

%
%

\newpage

\appendix

\input{appendix}

\end{document}

%% file: intro.tex
\section{Introduction}
This paper considers inference in game theoretic models of complete information. 
More precisely, we study the problem of recovering properties and characterizing parameter values
of the games that are consistent with observed equilibrium play, and provide a simple procedure based on convex optimization to recover both the region of consistent games and properties of said region, in a computationally efficient manner. Further, our approach has the power to compute the size of the region of consistent games, and hence to determine when approximate point identification of the true payoff matrices (or parameter values) is possible.


Our approach differs from most related work in that we depart from the usual distributional assumptions on the observer's knowledge of the unobserved variables and payoff shifters; instead, we adopt the weaker assumption that the unobserved variables belong to a known set. This 
increases the robustness of our approach,
reflecting ideas from the robust optimization literature---see~\cite{BGN2009,BBC2011,BB2012}. Furthermore, our approach is formally computationally efficient and, when implemented, is able to handle games of much larger size than those considered in previous work.


Our approach may be viewed as complementary to a model-driven approach, in that the tools we provide here may be used to objectively evaluate the quality of fit one achieves under certain modeling assumptions. Our approach also allows us to explore a variety of assumptions about the information that might available to an observer of game play, and the effects that these assumptions would have on constraining the space of consistent games.

%% file: summary_results.tex
\subsection{Summary of results}
We consider a setting where, at each step, a finite-action, 2-player game is played, and an observer observes a correlated equilibrium (a more permissive concept than Nash equilibrium) of the game.\footnote{Our framework extends to multi-player games with succinct representations; for clarity, we focus here on the two-player case. See Section~\ref{sec: succinct} for a discussion of succinct multi-player games.} We assume that the games played on each step are closely related, in that each reflects a small perturbation in the payoffs of some underlying game; such perturbations are often referred to as {\em{payoff shifters}}.


In a departure from previous work, we do not make distributional assumptions on the payoff shifters, nor do we make any assumption on how the players decide which equilibrium to play, when multiple equilibria are present. 
Instead, we assume that the observer knows nothing of the equilibrium selection rule, and that the information the observer has on the unobserved payoff shifters is simply that the unobserved payoff shifters belong to a known set (see Section~\ref{sec: model} for more details); this is a significantly weaker assumption than knowing exactly what distribution the payoff shifters are taken from. For example, imagine an analyst observes a routing game every day; the shifts in payoffs may come from a combination of several events such as changes in road conditions, traffic accidents, and work zones, whose potential effects on the costs of paths in a routing game may be difficult to predict and quantify precisely as a probability distribution.

In this setting, we give a computationally efficient characterization of the set of games that are consistent with the observations (Section~\ref{sec:consistent_set}); this set is ``sharp'', in the sense that it does not contain any game that is not consistent with the observations. One of our main new contributions is computational efficiency itself: the pioneering work of Beresteanu, Molchanov, and Molinari~\cite{BMM2011} only checks membership of a game to the set of consistent games, and does so in a manner that is tractable in small games but intractable for larger games---see Section~\ref{sec: related_work} for a more in-depth discussion. 
We also show that our framework accommodates an alternate model wherein the observer learns the expected payoff of each player at each equilibrium he sees; in our routing game example, think of an observer who sees the expected time each player spends in traffic. We refer to this setting as ``partial payoff information,'' and discuss it in Sections~\ref{sec: optim_framework} and~\ref{sec: payoff_info}.

Our second main contribution is our ability to quantify the size of the set of consistent games. We give an efficient algorithm (see Section~\ref{sec:checking_identifiability}, Algorithm~\ref{alg: alg_identification}) that takes a set of observations as  input and 
computes the diameter of the sharp region of consistent games. 
The diameter of the consistent set is of interest to an observer, because it gives him a measure of how sharp the conclusions he can draw from the observations are (the larger the diameter, the less sharp the conclusions), and in particular the diameter quantifies the level of approximate point identification that is achievable in a particular setting. Additionally, in Section~\ref{sec: payoff_info}, Lemmas~\ref{lem: d2-recovery} and \ref{lem: dinf-recovery}, we give structural conditions on the sets of observations that allow for accurate recovery.
We also exhibit examples in which said conditions do not hold, and therefore accurate recovery is not possible. 

We show we can extend our framework (Section~\ref{sec: linear_properties}) to find the set of consistent games when restricted to games with certain linear properties, e.g., zero-sum games, potential games, and games whose utilities can be parametrized by linear functions; this allows us to determine to what extent the observed behavior is consistent with such assumptions on the underlying game.

In Section~\ref{sec: extensions}, we show we can extend our framework to finite games with a large number of players, provided the game has a succinct representation. We further show our framework's potential to deal with games with infinite action spaces, using Cournot competition as an example.


Finally, in Section~\ref{sec:simulations}, we illustrate our approach with simulations, in both a simple entry game (Section~\ref{sim: entry}), and in large Cournot competition games (Section \ref{sim: cournot}).

%% file: related_work.tex
\subsection{Related work}\label{sec: related_work}
One important modeling issue is whether and why one would ever observe multiple, differing behaviors of a single agent. 
A natural, well-established approach models different observations found in the data as stemming from  random perturbations to the agents' utilities, as in~\cite{BV1984,BR1991,T2003,S2006,AL2010,AL2011,AL2012,BHKN2010,BHHR2011}. In dynamic panel models, one observes equilibria across several markets sharing common underlying parameters, and in particular~\cite{HM2010} considers a setting in which a unique, fixed equilibrium is played within each market. We adopt a similar approach here, and assume that we have access to several markets or locations that play perturbed versions of the same game, and that a single (mixed) equilibrium is played in each market.

In the branch of the econometrics literature that aims to infer parameters in game theoretic model, it is often the case that one requires that the game be small or that the utilities of the players can be written as simple functions of a restricted number of parameters. For example, 2-player entry games with entry payoffs parametrized as linear functions of a small number of variables, as seen in~\cite{T2003} and subsequent work, are among the most-studied in the literature. One drawback of this literature is that when the space of parameters is high-dimensional or when multiple equilibria exist, identification of the true parameters of the game often becomes impossible, since the observations do not correspond to a unique consistent explanatory game. A number of recent papers~\cite{ABJ2004,CT2009,BMM2011,NST2015} consider instead the problem of constructing {\em{regions}} of parameters that contain the true value of the parameters they aim to recover from equilibrium observations of games. For example, Nekipelov, Syrgkanis, and Tardos~\cite{NST2015} study a dynamic sponsored search auction game, and provide a characterization of the \textit{rationalizable set}, consisting of the set of private parameters that are consistent with the observations, under the relaxed assumption that players need not follow equilibrium play, but rather use some form of no-regret learning. Relatedly, Andrews, Berry, and Jia~\cite{ABJ2004} and Ciliberto and Tamer~\cite{CT2009} compute confidence regions for the value of the true parameter, but their regions are not ``sharp,'' in the sense that they may contain parameter values that are not consistent with some of the implications of their models.

Perhaps closest to the present work, Beresteanu, Molchanov, and Molinari~\cite{BMM2011} combine random set theory and convex optimization to give a representation of the \textit{sharp identification region} as the set of values for which the solution to a convex optimization program with a random objective function is almost surely (in the payoff shifters) equal to $0$.\footnote{Our notion of the {\em{consistent set}} is closely analogous to the sharp identification region of~\cite{BMM2011}. We use different terminology to highlight that they are derived under somewhat different settings.} Hence, verifying membership of a parameter value to the \textit{sharp identification region} can be done efficiently in simple settings such as entry-games with linearly parametrized payoffs. This is an exciting advance, especially when considering games with few players and small action sets; however, for computational reasons, the approach becomes impractical in large games, such as 2-player games with many actions per player:
\begin{itemize}
\item The Beresteanu, Molchanov, and Molinari~\cite{BMM2011} framework can verify that a vector of parameter values belongs to the sharp identification set, but does not provide an efficient, searchable representation of the sharp identification set itself, nor an algorithm to efficiently find a point in the set. 
\item One can verify that a parameter vector belongs to the sharp identification set by checking that a particular condition holds for almost all possible realizations of the payoff shifters. Beresteanu, Molchanov, and Molinari~\cite{BMM2011} further show that one can cluster payoff shifters into groups such that all perturbed games in the same group have the same set of Nash equilibria; one then must check the condition only once per group. In particular, in their entry-game example, the number of such groups is small, and thus this is a computationally tractable task. However, in more complex games, or for more general equilibrium concepts, the number of such groups can become intractably large in the number of actions available to each player.
\item Finally, the BMM framework~\cite{BMM2011} relies on being able to compute {\em{all}} equilibria of each of the perturbed games, which may be impractical. (In general, in the worst case, it is computationally hard to find \emph{any} Nash equilibrium in time polynomial in the number of actions of each player, even for  2-player games~\cite{CDT2007}.)

\end{itemize}

The goal of the present paper is similar to the goal of~\cite{BMM2011}, in the sense that we wish to sharply understand the set of games that are consistent with a set of observations (for us, correlated equilibria of perturbed games). We also use the setting of their simulations as the jumping off point for our own experimental section. However, our approach differs from that of~\cite{BMM2011} in two main ways. First, we make weaker assumptions on the information on the unobserved payoff shifters available to the observer; our approach to modeling the perturbations is inspired by the concept of uncertainty sets in robust optimization (see~\cite{BGN2009,BBC2011,BB2012}). 
Second, our framework provides a computationally efficient characterization of the consistent set, both in theory and in practice, on games of large size, and also gives efficient and practical algorithms to find points in the consistent set, compute its diameter, and test whether it contains games with certain properties. 

A handful of papers in the computer science literature have looked at slightly different but related questions, arising when observing equilibria of games whose payoffs are unknown. In particular, Bhaskar et al.~\cite{BLSS2014} and Rogers et al.~\cite{RRUW2015} study a network routing setting in which equilibrium behavior can be observed but edge costs are unknown, and study the query complexity of devising a variant of the game to induce desired target flows as equilibria. Barman et al.~\cite{BBEW2013} adopt a model in which the observer observes what joint strategies are played when restricting the actions of the players in a complete information game with no perturbations, and show that data with certain properties can be rationalized by games with low complexity. Perhaps closest to our work,~\cite{KS2015} gives a convex and computationally efficient characterization of the set of consistent games when the true game is succinct and its structure is known; however, their setting does not incorporate payoff shifters, instead assuming they observe different equilibria from different succinct games with different structures, but that all these games share the same parameter value.

%% file: model.tex
\section{Model and setting}\label{sec: model}

\subsection{Players' behavior}
Consider a finite two-player game $G$; we will refer to it as the {\em true} or {\em underlying} game. Let $\cA_1,\cA_2$ be the finite sets of actions available to players $1$ and $2$, respectively, and let $m_1=| \cA_1 |$ and $m_2=| \cA_2 |$ be the number of actions available to them. For every $(i,j) \in \cA_1 \times \cA_2$, we denote by $G_p(i,j)$ the payoff of player $p$ when player $1$ chooses action $i$ and player $2$ chooses action $j$. $G_p \in \RR^{m_1 \times m_2}$ is the vector representation of the utility of player $p$, and we often abuse notation and write $G=(G_1,G_2)$. The {\em strategies} available to player $p$ are simply the distributions over $\cA_p$. 
A {\em strategy profile} is a pair of strategies (distributions over actions), one for each player. A {\em joint strategy profile} is a distribution over pairs of actions (one for each player); it is not required to be a product distribution.
We refer to strategies as {\em pure} when they place their entire probability mass on a single action, and {\em mixed} otherwise.

We consider $l$ perturbed versions of the game $G$, indexed by $k \in [l]$ so that the $k^{th}$ perturbed game is denoted $G^k $; one can for instance imagine each $G^k$ as a version of the game $G$ played in a different location or market $k$. The same notation as for $G$ applies to the $G^k$'s.

Throughout the paper, we assume that for each $k$, the players' strategies are given by a correlated equilibrium of the complete information game $G^k$. In the presence of several such equilibria, no assumption is made on the selection rule the players use to pick which equilibrium to play (though we assume they both play according to the same equilibrium). Correlated equilibria are defined as follows:

\begin{definition}
A probability distribution $e$ is a \textit{correlated equilibrium} of game $G=(G_1,G_2)$ if and only if
\begin{align*}
\sum\limits_{j=1}^{m_2} G_1(i,j) e_{ij} \geq \sum\limits_{j=1}^{m_2} G_1(i',j) e_{ij} \; \forall i,i' \in \cA_1
\\\sum\limits_{i=1}^{m_1} G_2(i,j) e_{ij} \geq \sum\limits_{i=1}^{m_1} G_2(i,j') e_{ij} \; \forall j,j' \in \cA_2
\end{align*}
\end{definition}
The notion of correlated equilibrium extends the classical notion of Nash equilibrium by allowing players to act jointly; as every Nash equilibrium of a game is a correlated equilibrium of the same game, many of our results also have implications for Nash equilibria.

\subsection{Observation model}
Our observer does not have access to the payoffs of the underlying game $G$ nor of the perturbed games $G^k$, for any $k$ in $[l]$. We model an observer as observing, for each perturbed game $G^k$, the entire correlated equilibrium distribution $e^k \in \RR^{m_1 \times m_2}$, where $e^k(i,j)$ denotes the joint probability in the $k$th perturbed game of player $1$ playing action $i \in \cA_1$ while player $2$ plays action $j \in \cA_2$.\footnote{The reader may interpret this assumption as describing a situation in which each perturbed game is played repeatedly over time, with the same (possibly mixed) equilibrium played each time, allowing the observer to infer the probability distribution over actions that is followed by the players. Technically, using samples to estimate an empirical distribution over actions would yield a $\varepsilon$-approximate equilibrium for some $\varepsilon$ that depends on the sample size. Our framework can trivially be extended to deal with such approximate correlated equilibria; for simplicity, we omit further discussion of this issue.} Note that as $e^k$ represents a probability distribution, we require $e^k(i,j) \geq 0 \; \forall (i,j)$ and $\sum\limits_{i,j} e^k_{ij}=1$.
In this paper, we consider two variants of the model of observations we just described:
\begin{itemize}
\item In the \textit{partial payoff information} setting, the observer has access to equilibrium observations $e^1,\ldots,e^l$, and additionally to the expected payoff of equilibrium $e^k$ on perturbed games $G^k$, for each player $p$ and for all $k \in [l]$; we denote said payoff $v_p^k$ and note that $v_p^k=e^{k \; \prime} G_p^k$. 
\item In the \textit{payoff shifter information} setting, at each step $k$, a \textit{payoff shifter} $\beta^k=(\beta_1^k,\beta_2^k) \in \cR^{m_1 \times m_2 } \times \cR^{m_1 \times m_2 }$ is added to game $G=(G_1,G_2)$, and the perturbed games $G^k$ result from the further addition of small perturbations to the $G+\beta^k$'s. The observer knows $\beta^1,\ldots,\beta^l$ and observes $e^1,\ldots,e^l$ of perturbed games $G^1,\ldots,G^l$. This  setting represents a situation in which changes in the behavior of agents are observed as a function of changes in observable economic parameters (taxes, etc.). 
\end{itemize}

While the payoff shifter information setting is the model of perturbations that is commonly used in the literature, the partial payoff information setting has not been used in previous work, to the best of our knowledge. We introduce it to model the following types of situation: 

\begin{example}
Two firms are competing for customers in Los Angeles, and an observer follows what actions the two L.A. firms take over the course of each quarter. The observer also learns the quarterly revenue of each firm.
\end{example}

\begin{example}
Several agents are playing a routing game, and the observer sees not just the routes the players choose, but also the amount of time players spend in traffic.
\end{example}

\subsection{Observer's knowledge about the perturbations}
Our paper aims to characterize the games that explain equilibrium observations under the partial payoff and  payoff shifter information settings when the perturbations are known to be ``small'' and the perturbed games are thus ``close'' to the underlying game. The next few definitions formalize our notion of closeness, and Assumption~\ref{assmpt:perturbations} formalizes the information the observer has about the perturbations added to the underlying game $G$.

\begin{definition}
A game $G$ is $\delta$\textit{-close} to games $G^1,\ldots,G^l$ with respect to metric $d$ for $\delta > 0$ if and only if $d(G^1,\ldots,G^l|G) \leq \delta$.
\end{definition}

We think of $d$ as distances and therefore convex functions of the perturbations $G-G^k$ for all $k$. For the above definitions to make sense in the context of this paper, we need a metric whose value on a set of games $G,G^1,\ldots,G^l$ is small when $G,G^1,\ldots,G^l$ are close in terms of payoffs. 
We consider the following metrics:

\begin{definition}
The {\em{sum-of-squares distance}} between games $G$ and $G^1,\ldots,G^l$ is given by 
$$d_2(G^1,\ldots,G^l|G)=\sum\limits_{k=1}^l (G_1-G^k_1)'(G_1-G^k_1)+\sum\limits_{k=1}^l (G_2-G^k_2)'(G_2-G^k_2).$$
The {\em{maximum distance}} between games $G$ and $G^1,\ldots,G^l$ is defined as
$$d_{\infty}(G^1,\ldots,G^l|G)=\max\limits_{p \in \{1,2\}, \; k \in [l]} \Vert G_p-G^k_p \Vert_{\infty},$$
where $\Vert . \Vert_{\infty}$ denotes the usual infinity norm.
\end{definition}

Both distances are useful, in different situations. The sum-of-squares distance is small when the variance of the perturbations added to $G$ is known to be small, but allows for worst-case perturbations to be large. An example is when the $G^k$'s are randomly sampled from a distribution with mean $G$, unbounded support, and small covariance matrix, in which case some of the perturbations may deviate significantly from the mean but with low probability, while the average squared perturbation remains small. If the distribution of perturbations is i.i.d Gaussian, the sum-of-squares norm replicates the log-likelihood of the estimations and follows a Chi-square distribution. The maximum distance, in contrast, is small when it is known that all perturbations are small and bounded; one example is when the perturbations are uniform in a small interval $[-\delta,\delta]$. 



Throughout the paper, we make the following assumption on the information about the perturbations that is available to the observer:
\begin{assumption}\label{assmpt:perturbations}
Let $G$ be the underlying game and $G^1,\ldots, G^l$ be the perturbed games that generated observations $e^1,\ldots,e^l$.
\begin{itemize}
\item In the partial payoff information settings, the observer knows that $G$ is $\delta$\textit{-close} to games $G^1,\ldots,G^l$ with respect to some metric $d$ and magnitude $\delta \geq 0$.
\item In the payoff shifter information setting with observed shifters $\beta^1,\ldots,\beta^l$, the observer knows that $G$ is $\delta$\textit{-close} to the unshifted games $(G^1-\beta^1),\ldots,(G^l-\beta^l)$ with respect to some metric $d$ and magnitude $\delta \geq 0$.
\end{itemize}
\end{assumption}

Assumption~\ref{assmpt:perturbations} defines a convex set the observer knows the perturbations must belong to, much like the uncertainty sets given in~\cite{BB2012,BGN2009,BBC2011}. We note that the $d_2$ and $d_{\infty}$ distances we focus on define respectively an ellipsoidal and a polyhedral uncertainty set (as seen in~\cite{BBC2011}).

\begin{remark}
While we make Assumption~\ref{assmpt:perturbations} for convenience and simplicity of exposition, our framework is able to handle more general sets of perturbations. In particular, the results of Section~\ref{sec: optim_framework} can easily be extended to any convex set of perturbations that has an efficient, easy-to-optimize-over representation. This includes classes of sets defined by a tractable number of linear or convex quadratic constraints, which in turn encompasses many of the uncertainty sets considered in~\cite{BB2012}, such as the central limit theorem or correlation information sets, and most of the typical sets presented therein.
\end{remark}


\subsection{Consistent games}In this paper, as in~\cite{BMM2011}, we adopt an observation-driven view that describes the class of games that are consistent with the observed behavior. Given a set of observations, we define the \textit{set of consistent games} as follows:

\begin{definition}[$\delta$-consistency]\label{def:consistency}
We say a game $\tilde{G}$ is {\em{$\delta$-consistent}} with the observations when there exists a set of games $(\tilde{G},\tilde{G}^1,\ldots,\tilde{G}^l)$ such that for all $k$, $e^k$ is an equilibrium of $\tilde{G}^k$, and:
\begin{itemize}
\item If in the partial payoff information model of observations, $\sum_p \tilde{G}_p^{k \; \prime} e^k = v_p^k$ for all players $p$ and $d(\tilde{G}^1,\ldots,\tilde{G}^l|\tilde{G}) \leq \delta$. 
\item If in the payoff shifter information model, $d(\tilde{G}^1-\beta^1,\ldots,\tilde{G}^l-\beta^l|\tilde{G}) \leq \delta$.
\end{itemize}
The set of all $\delta$-consistent games with respect to metric $d$ is denoted $S_{d}(\delta)$.
\end{definition}

Given the specifications of our model, it is often the case that, given a set of observations with no additional assumption on the distribution of perturbations nor on a the rule used to select among multiple equilibria, it is not possible to recover an approximation to a unique game that generated these observations (no matter what recovery framework is used). That is, the diameter of the consistent set can sometimes be too large for approximate point identification to be possible, which is highlighted in the following example:
\begin{example}\label{ex: trivial_explanation}
Take any set of  
observations $e^1,\ldots,e^l$ under the no payoff information observation model, and let $\hat{G}$ be the all-constant game, i.e., $\hat{G}_1(i,j)=\hat{G}_2(i,j)=c$ for some $c \in \mathbb{R}$ and for all $(i,j) \in \cA_1 \times \cA_2$. Let $\hat{G}^1=\ldots=\hat{G}^l=\hat{G}$. Then for all $k \in [l]$, $e^k$ is an equilibrium of $\hat{G}^k$, and $d_2(G^1,\ldots,G^l|G)=d_{\infty}(G^1,_ldots,G^l|G)=0$. That is, $\hat{G}$ is a trivial game, and it is consistent with all possible observations. Even when $e^1,\ldots,e^l$ are generated by a non-trivial $G$, without any additional observations, an observer cannot determine whether $G$ or $\hat{G}$ is the underlying game.  In fact, both games are consistent with all implications of our model. We note that this issue arises regardless of how inferences will be drawn about the observations, so long as the approach does not discard consistent games.
\end{example}
It may thus be of interest to an observer to compute the diameter of the consistent set, either to determine whether point identification is possible, or simply to understand how tightly the observations constrain the space of consistent games. We define it as follows:
\begin{definition}\label{def: diameter}
The diameter $D(S_d(\delta))$ of consistent set $S_d(\delta)$ is given by 
$$D(S_d(\delta))=\sup \{ \Vert \hat{G} - \tilde{G} \Vert_{+\infty} \text{ s.t. } \tilde{G},\hat{G} \in S_d(\delta) \}$$
\end{definition}
When the diameter is small, then every game in the consistent set is close to the true underlying game, and approximate point identification is achievable. When the diameter $D(S_d(\delta))$ grows large, point identification is impossible independently of what framework is used for recovery, as there exist two games that are $D(S_d(\delta))$-far apart in terms of payoff, yet either could have generated all observations.

%% file: convex_framework.tex
\section{A convex optimization framework}\label{sec: optim_framework}
In this section, we show how techniques from convex optimization can be used to 
recover the perturbation-minimizing explanation for a set of observations, determine the extent to which observations are consistent with certain assumptions on the underlying game, and determine whether a set of observations tightly constrains the set of games that could explain it well. The results in this section are not tied to a specific observation model.

\subsection{Efficient characterization of the set of consistent games}\label{sec:consistent_set}

We will show that for every $\delta$, and $d \in \{d_2,d_{+\infty}\}$, the set of consistent games $S_d(\delta)$ has an efficient, convex representation.
\begin{claim}
If in the ``partial payoff information'' model of observations:
\[
S_{d}(\delta)= \left \{ 
 G \text{ s.t } 
\begin{array}{ll@{}ll} & \exists (G^1,\ldots,G^l) \text{ with } d(G^1,\ldots,G^l|G) \leq \delta \text{ s.t. }\\
		& \sum\limits_{j=1}^{m_2} G^k_1(i,j) e^k_{ij} \geq \sum\limits_{j=1}^{m_2} G^k_1(i',j) e^k_{ij} \; \forall i,i' \in \cA_1, \forall k \in [l],\\
                 	& \sum\limits_{i=1}^{m_1} G^k_2(i,j) e^k_{ij} \geq \sum\limits_{i=1}^{m_1} G^k_2(i,j') e^k_{ij} \; \forall j,j' \in \cA_2, \forall k \in [l] \\
		& \sum_p \tilde{G}_p^{k \; \prime} e^k = v_p^k
\end{array}
\right \}
\]
If in the ``payoff shifter information'' model:
\[
S_{d}(\delta)= \left \{ 
 G \text{ s.t } 
\begin{array}{ll@{}ll} & \exists (G^1,\ldots,G^l) \text{ with } d(\tilde{G}^1-\beta^1,\ldots,\tilde{G}^l-\beta^l|\tilde{G}) \leq \delta \text{ s.t. }\\
		& \sum\limits_{j=1}^{m_2} G^k_1(i,j) e^k_{ij} \geq \sum\limits_{j=1}^{m_2} G^k_1(i',j) e^k_{ij} \; \forall i,i' \in \cA_1, \forall k \in [l],\\
                 	& \sum\limits_{i=1}^{m_1} G^k_2(i,j) e^k_{ij} \geq \sum\limits_{i=1}^{m_1} G^k_2(i,j') e^k_{ij} \; \forall j,j' \in \cA_2, \forall k \in [l]	
\end{array}
\right \}
\]
\end{claim}
\begin{proof}
Follows from the definion of $\delta$-consistency (Definition~\ref{def:consistency})
\end{proof}
We remark that as in~\cite{BMM2011}, our sets are \textit{sharp}: any game that explains the observations belongs to this set, and any game that belongs to this set is consistent with our assumptions and observations. Indeed, if $G$ is in the consistent set, there must exist perturbations of valid magnitude (given by the corresponding $G^1,\ldots,G^l$) and an equilibrium $e^k$ of each perturbed game that together would lead to our observations, by the definition of the consistent set. 

These consistent sets have efficient convex representations, for two reasons. First, all constraints are always linear except those of the form 
\\$d(G_p^1,\ldots,G_p^l|G) \leq \delta$. When $d=d_2$, $d(G_p^1,\ldots,G_p^l|G) \leq \delta$ is a simple convex quadratic constraint, while when $d=d_{\infty}$, $d(G_p^1,\ldots,G_p^l|G) \leq \delta$ is equivalent to the following collection of linear constraints: 
\begin{align*}
-\delta \leq G_p^k-G_p \; \forall p \in \{1,2\}, \forall k \in [l]\\
 G_p^k-G_p \leq \delta \; \forall p \in \{1,2\}, \forall k \in [l].
\end{align*}
Second, the number of constraints describing each set is quadratic in the number of player actions $m_1$ and $m_2$.

As mentioned in Section~\ref{sec: model}, in all observation models, the assumption that $d(G_p^1,\ldots,G_p^l|G) \leq \delta$ can easily be replaced by an assumption of the perturbations $G^k-G$ being in any tractable convex set. In particular, many of the sets considered in~\cite{BB2012} fit this requirement, and they describe robust information that an observer without distributional knowledge of the perturbations could realistically have on said perturbations: for example, an observer could know that the sum or average of the perturbations satisfies certain lower- and upper-bounds.

\subsection{Recovering the perturbation-minimizing consistent game}\label{sec:basic_optimization}

Here, we consider the problem of recovering a game that best explains a given set of observations from perturbed games, according to the desired distance metric $d$. One reason to do so is that it enables an observer to test whether there exists any game in $S_d(\delta)$ that is consistent with specific properties and to give a measure of how much of $S_d(\delta)$ has said properties---see Section~\ref{sec: linear_properties}. Or, it could be that the observer is simply interested in recovering the ``best'' game according to any simple convex metric of interest. For any metric $d$ and any observation model, this can be done simply by solving:
\begin{align}
\begin{array}{ll@{}ll}\label{optim_unseparated}
\min\limits_{G,\delta}  & \displaystyle \delta &\\
\text{s.t.}	& G \in S_d(\delta)	
\end{array}
\end{align}
It is easy to see that this program returns the game $G$ and the minimum value of $\delta$ such that $d(G^1,\ldots,G^l|G) \leq  \delta$ (resp. $d(G^1-S^1,\ldots,G^l-S^l|G) \leq  \delta$) in the partial payoff information setting (resp. the payoff shifter information setting) where the $G^k$'s satisfy all equilibrium constraints, hence Program~\eqref{optim_unseparated} returns the perturbation-minimizing $G$ that is consistent with all observations. When $d=d_{\infty}$, this is a linear program; when $d=d_2$, this is a second-order cone program, using the same reasoning as in Section~\ref{sec:consistent_set} (this holds even with $\delta$ as a variable). Both types of programs can be solved efficiently, as seen in~\cite{boyd2004convex}.

\subsection{Can observations be explained by linear properties?}\label{sec: linear_properties}
This convex optimization-based approach can further be used to determine whether there exists a game that is compatible with the observations and that also has certain additional properties, as long as the properties of interest can be expressed as a tractable number of linear equalities and inequalities. One can then solve program~\eqref{optim_unseparated} with said linear equalities and inequalities as additional constraints (the program remains a SOCP or LP with a tractable number of constraints), then check whether the optimal value is greater than or less than $\delta$. If the optimal value is greater than $\delta$, then there exists no game with those properties that belongs to the $\delta$-consistent set; if the optimal value is smaller than $\delta$, then the recovered game displays the additional properties and  belongs to the $\delta$-consistent set. In what follows, we present a few examples of interesting properties that fit this framework.

\subsubsection{Zero-sum games}

A zero-sum game is a game in which for each pure strategy $(i,j)$, the sum of the payoff of player $1$ and the payoff of player $2$ for $(i,j)$ is always $0$. One can restrict the set of games we look for to be zero-sum games, at the cost of separability of Program~\eqref{optim_unseparated}, by adding constraints $G_1(i,j)=-G_2(i,j) \; \forall (i,j) \in \cA_1 \times \cA_2$.

\subsubsection{Exact potential games}

A 2-player game $G$ is an exact potential game if and only if it admits an exact potential function, i.e., a function $\Phi$ that satisfies: 
\begin{align}
\Phi(i,j)-\Phi(i',j)=G_1(i,j)-G_1(i',j) \; \forall i,i' \in \cA_1, \forall j \in \cA_2 \label{eq: potential_constraint_1}
\\ \Phi(i,j)-\Phi(i,j')=G_2(i,j)-G_2(i,j') \; \forall i \in \cA_1, \forall j,j' \in \cA_2 \label{eq: potential_constraint_2}
\end{align}
In order to restrict the set of games we are searching over to the set of potential games, one can introduce $m_1 m_2$ variables $\Phi(i,j)$ and constraints~\eqref{eq: potential_constraint_1},~\eqref{eq: potential_constraint_2} in Program~\eqref{optim_unseparated}. 
%

\subsubsection{Games generated through linear parameter fitting}
It is common in the literature to recover a game with the help of a parametrized function whose parameters are calibrated using the observations. In many applications, linear functions of some parameters are considered---entry games are one example. Our framework allows one to determine whether there exist parameters for such a linear function that provide good explanation for the observations. When such parameters exist, one can use the mathematical program to find a set of parameters that describe a game which is consistent with the observations.
Take two functions $f_1(\theta)$ and $f_2(\theta)$ that are linear in the vector of parameters $\theta$ and output a vector in $\mathbb{R}^{m_1 \times m_2}$. It suffices to add the optimization variable $\theta$ and the linear constraints $G_1=f_1(\theta)$ and $G_2=f_2(\theta)$ to Program~\eqref{optim_unseparated} to restrict the set of games we look for to games linearly parametrized by $f_1,f_2$.

%% file: diameter.tex
\subsection{Computing the diameter of the consistent set}\label{sec:checking_identifiability}

In this section, we provide an algorithm (Algorithm~\ref{alg: alg_identification}) for computing the diameter of $S_d(\delta)$, for any given value of $\delta$. Because the diameter is a property of the consistent set and not of the framework used to recover an element from said set, this tells an observer whether approximate point identification is possible \textit{independently} of what framework is used for recovery. In particular, when the diameter is small, our framework approximately recovers the true underlying game (see Section~\ref{sec:basic_optimization}). When the diameter is large, \textit{no} framework can achieve approximate point identification of a true, underlying game.

\begin{algorithm}[!h]
\SetAlgoNoLine
\KwIn{Observations $e^1,\ldots,e^l$, magnitude of perturbations $\delta$, metric $d$}
\KwOut{Real number $\mathcal{A}(\delta)$ (may be infinite)}
\For{$(i,j) \in \cA_1 \times \cA_2$, player $p \in {1,2}$}
{\begin{equation*}
\begin{array}{ll@{}ll}
P_{\delta,p}(i,j)= & \sup\limits_{\tilde{G},\hat{G},\gamma} \; & \displaystyle \gamma \\
& \text{s.t.} & \tilde{G} \in S_d(\delta)\\
		&&\hat{G} \in S_d(\delta)\\
		&& \tilde{G}_1(i,j)-\hat{G}_1(i,j) \geq \gamma\\
\end{array}
\end{equation*}
}

$\mathcal{A}(\delta)=\max\limits_{(i,j)\in \cA_1 \times \cA_2} \max\limits_{p \in \{1,2\}} P_{\delta,p}(i,j)$

\caption{Computing the diameter of the consistent set}
\label{alg: alg_identification}
\end{algorithm}
Algorithm~\ref{alg: alg_identification} is computationally efficient for the considered metrics $d_2$ and $d_{\infty}$: it solves $2 m_1 m_2$ linear programs for $d_\infty$, and $2 m_1 m_2$ second-order cone programs (SOCP) for $d=d_2$ with a tractable number of constraints. The algorithm computes the diameter of the consistent set:
\begin{lemma}\label{properties_algo}
The output $\mathcal{A}(\delta)$ of Algorithm~\ref{alg: alg_identification} run with input $\delta$ satisfies $\mathcal{A}(\delta) = D(S_d(\delta))$.
\end{lemma}
\begin{proof}
See Appendix~\ref{sec: diameter_proof}.
\end{proof}


%% file: payoff_info_recovery.tex
\section{Consistent games with partial payoff information: when is recovery possible?}\label{sec: payoff_info}

This section considers the \textit{partial payoff information} variant of the observation model described in Section~\ref{sec: model}. We ask the following question: when is it possible to approximate the underlying game, in the presence of partial payoff information? We answer this question by giving bounds on the diameter of the consistent set $S_d(\delta)$ as a function of $\delta$ and the observations $e^1,...,e^l$, for both metrics $d_2$ and $d_\infty$.

Recall that in this setting, for an equilibrium $e^k$ observed from perturbed game $G^k$, the observer learns not only $e^k$, but also the expected payoff  $v_p^k$ of player $p$ in said equilibrium strategy on game $G^k$. Similar to the previous sections, we are interested in computing  a game $\hat{G}$ that is close to some perturbed games $\hat{G}^1,...,\hat{G}^l$ that (respectively) have equilibria $e^1,...,e^l$ with payoffs $v^1,...,v^l$. For simplicity of presentation, we recall that the optimization program that the observer solves is separable and note that he can thus solve the following convex optimization problem for player 1, and a similar optimization problem for player 2:
\begin{equation}\label{primal_program_payoff}
\begin{array}{ll@{}ll}
&P(\epsilon)&=\min\limits_{G_1^k,G_1} \; \displaystyle d(G_1^1,...,G_1^k|G_1) &\\
& \text{s.t.}  & \sum\limits_{j=1}^d G_1^k(i,j) e^k_{ij} \geq \sum\limits_{j=1}^d G_1^k(i',j) e^k_{ij} \; \forall i,i' \in \cA_1, \forall j \in \cA_2, \forall k \in [l]\\
		&&e^{k \; \prime} G_1^k = v_1^k \;  \forall k \in [l]
\end{array}
\end{equation}


We take $l \geq m_1 m_2$ and make the following assumption for the remainder of this subsection,
unless otherwise specified:
\begin{assumption}
There exists a subset $E \subset \{e^1,...,e^l\}$ of size $m_1 m_2$ such that the vectors in $E$ are linearly independent.
\end{assumption}
We abuse notation and denote by $E$ the $m_1 m_2 \times m_1 m_2$ matrix in which row $i$ is given by the $i^{th}$ element of set $E$, for all $i \in [m_1 m_2]$; also, we write $d(G^1,...,G^l|G)=\sum\limits_p d(G_p^1,...,G_p^l|G_p)$, i.e., $d(G_p^1,...,G_p^l|G_p)$ is the part of $d(G^1,...,G^l|G)$ that corresponds to player $p$. For every $p \in \mathbb{N} \cup \{ +\infty \}$, let $\Vert . \Vert_p$ be the p-norm. We can define the corresponding induced matrix norm $\Vert . \Vert_p$ that satisfies $\Vert M \Vert_p= \sup\limits_{x \neq 0} \frac{\Vert M x \Vert_p}{\Vert x \Vert_p}$  for any matrix $M \in \mathbb{R}^{m_1 m_2 \times m_1 m_2}$.

Lemmas~\ref{lem: d2-recovery} and~\ref{lem: dinf-recovery} highlight that if one has $m_1 m_2$ linearly independent observations (among the $l$ equilibrium observations) such that the induced matrix of observations $E$ is well-conditioned, and the perturbed games are obtained from the underlying game through small perturbations, any optimal solution of Program~(\ref{primal_program_payoff}) necessarily recovers a game whose payoffs are close to the payoffs of the underlying game. The statements are given for both metrics introduced in Section~\ref{sec: model}.

\begin{lemma}\label{lem: d2-recovery}
Let $G$ be the underlying game, and $G^1, ...,G^{l}$ be the games generating observations $e^1,..,e^{l}$, where $l=m_1 m_2$. Suppose that for player $p$, $d_2(G_p^1,...,G_p^l|G_p)   \leq \delta$. Let ($\hat{G}_p,\hat{G}_p^1,...,\hat{G}_p^l)$ be an optimal solution of~Program (\ref{primal_program_payoff}) for player $p$ with distance function $d_2$. Then  
$$\Vert G_p-\hat{G}_p \Vert_2 \leq  \sqrt{2 \Vert E^{-1} \Vert_2 \cdot \delta}.$$
\end{lemma}

\begin{proof}
For simplicity of notation, we drop the $p$ indices. We first remark that $(G,G^1,...,G^l)$ is feasible for Program~\eqref{primal_program_payoff}; as $(\hat{G},\hat{G}^1,...,\hat{G}^l)$ is optimal, it is necessarily the case that 
$$\sum\limits_{k=1}^{l} \Vert \hat{G}-\hat{G}^k \Vert_2^2   \leq \sum\limits_{k=1}^{l}  \Vert G-G^k \Vert_2^2  \leq \delta.$$ 
Let us write $\Delta G=G-\hat{G}$. We know that for all $k$, $e^{k \; \prime} G^k=e^{k \; \prime} \hat{G}^k=v^k$, and thus $e^{k \; \prime} (G^k -\hat{G}^k)=0$. We can write 
\begin{align*}
E \Delta G&=(e_1'(G-\hat{G}) \; ... \; e_l'(G-\hat{G}))' \\
& =(e_1'(G-G^1+G^1-\hat{G}^1+\hat{G^1}-\hat{G}) \; ... \; e_l'(G-G^l+G^l-\hat{G}^l+\hat{G^l}-\hat{G}))'
\\& =(e_1'(G-G^1+\hat{G^1}-\hat{G}) \; ... \; e_l'(G-G^l+\hat{G^l}-\hat{G}))'.
\end{align*}
Let $x_k=G-G^k+\hat{G^k}-\hat{G}$. We then have $\Vert E \Delta G \Vert_2^2 =\sum\limits_{k=1}^l x_k' e_k e_k' x_k$, as $e_k e_k'$ is a symmetric, positive semi-definite, stochastic matrix, all its eigenvalues are between $0$ and $1$ and 
$$ \Vert E \Delta G \Vert_2^2 \leq \sum\limits_{k=1}^l x_k' x_k= \sum\limits_{k=1}^l \Vert x_k \Vert_2^2 \leq 2 \delta.$$
It immediately follows that $\Vert \Delta G \Vert_2 \leq \sqrt{2 \Vert E^{-1} \Vert_2 \cdot \delta}.$
\end{proof}


\begin{lemma}\label{lem: dinf-recovery}
Let $G$ be the underlying game, and $G^1, ...,G^{l}$ be the games generating observations $e^1,..,e^{l}$, where $l=m_1 m_2$. Suppose that for player $p$, $d_{\infty}(G_p^1,...,G_p^l|G_p)   \leq \delta$. Let ($\hat{G}_p,\hat{G}_p^1,...,\hat{G}_p^l)$ be an optimal solution of~Program~(\ref{primal_program_payoff}) for player $p$ with distance function $d_{\infty}$. Then  
$$\Vert G_p-\hat{G}_p \Vert_{\infty} \leq 2 \Vert E^{-1} \Vert_{\infty}  \cdot \delta.$$
\end{lemma}


\begin{proof}
See Appendix~\ref{infinite_norm_proof}.
\end{proof}

When $E$ is far from being singular, as long as the perturbations are small, we can accurately recover the payoff matrix of each player. This has a simple interpretation: the further $E$ is from being singular, the more diverse the observations are. More diverse observations means more information for the observer, and allows for more accurate recovery. On the other hand, if the matrix $E$ is close to being singular, it means that one sees the same or similar observations over and over again, and does not gain much information about the underlying game---there are more payoffs to recover than different observations, and the system is underdetermined. 

An extreme example arises when we take $E$ to be the identity matrix, in which case we observe every single pure strategy of the game and an approximation of the payoff of each of these strategies, allowing us to approximately reconstruct the game. It is also the case that there are examples in which $\Vert E^{-1} \Vert_{\infty}$ is large and there exist two games that are far from one another, yet both explain the observations, making our bound essentially tight:

\begin{example}
Consider the square matrix $E \in \mathbb{R}^{4 \times 4}$ with probability $0.25+\epsilon$ on the diagonal and $\frac{0.75-\epsilon}{3}$ off the diagonal, i.e., we observe four equilibria, each placing probability slightly higher than $0.25$ on a different action profile; the first equilibrium has a higher probability on action profile (1,1), the second on (1,2), the third on (2,1) and the last one on (2,2). Suppose the vector of observed payoffs is $v=(\delta, -\delta, \delta, -\delta)$, where $v(i)$ is the payoff for the $i^{th}$ equilibrium. Note that there exists a constant $C$ such that for all $\epsilon > 0$ small enough, $\Vert E^{-1} \Vert_{+\infty} \leq \frac{C}{\epsilon}$. 

In the rest of the example, we fix the payoff matrix of player $2$ for all considered games to be all-zero so that it is consistent with every equilibrium observation, and describe a game through the payoff matrix of player $1$. Let $G$ be the all-zero game, $G^1=G^3$ be the game with payoff $\frac{\delta}{0.5+2\epsilon/3}$ on actions (1,1) and (1,2) and 0 everywhere else, and $G^2=G^4$ be the game with payoff $-\frac{\delta}{0.5+2\epsilon/3}$ on actions (2,1) and (2,2) and 0 everywhere else. The $G^i$'s are consistent with the payoff observations as the payoffs are constant across rows on the same column, making no deviation profitable, and the payoff of equilibria $e^1$ and $e^3$ on $G^1$ and $G^3$ is indeed $\delta$, and $-\delta$ for $e^2$ and $e^4$ on $G^2$ and $G^4$. We have 
$$d_{\infty}(G^1,G^2,G^3,G^4|G)=\frac{\delta}{0.5+2\epsilon/3} \leq 2\delta$$
and 
$$\lim_{\epsilon \to 0} d_{\infty}(G^1,G^2,G^3,G^4|G) = 2 \delta.$$

Now, take $\hat{G}$ to be the game that has payoff $\delta/\epsilon$ for action profiles (1,1) and (1,2), and $-\delta/\epsilon$ for (2,1) and (2,2). Take $\hat{G}^1=\hat{G}^3$ to be the game with payoffs $\frac{\delta}{\epsilon}$ in the first column, and $-\frac{\delta}{\epsilon} \frac{3-2\epsilon}{3-4\epsilon}$ in the second column; similarly, take $\hat{G}^2=\hat{G}^4$ to be the game with payoffs  $\frac{\delta}{\epsilon} \frac{3-2\epsilon}{3-4\epsilon}$ in the first column and $-\frac{\delta}{\epsilon}$ in the second column. The observations are equilibria of the $\hat{G}^i$'s and yield payoff $\delta$. Now, note that for $\epsilon < 3/4$,
$$d_{\infty}(G^1,G^2,G^3,G^4|G)=\frac{\delta}{\epsilon} \left | 1- \frac{3-2\epsilon}{3-4\epsilon} \right | = \frac{2}{3-4\epsilon} \delta.$$

Therefore, both $G$ and $\hat{G}$ are good explanations of the equilibrium observations, in the sense that for $\epsilon \leq 1/4$, $G$ is $\delta$-close to $G^1,...,G^l$ and $\hat{G}$ is $\delta$-close to $\hat{G}^1,...,\hat{G}^l$ that have $e^1,...,e^l$ as equilibria, respectively. However, 
$$\Vert G - \hat{G} \Vert_{\infty}=\frac{\delta}{\epsilon}-\frac{\delta}{0.5+2\epsilon/3} \geq \delta \left(\frac{1}{\epsilon}-2\right),$$
which immediately implies
$$\Vert G - \hat{G} \Vert_{\infty} = \Omega_{\epsilon \to 0}\left(\frac{\delta}{\epsilon}\right)= \Omega_{\epsilon \to 0}\left(\Vert E^{-1} \Vert_{\infty} \delta\right).$$
\end{example}


\begin{remark}
In the case of {\em sparse} games, in which some action profiles are never profitable to the players, and are therefore never played, one can reduce the number of linearly independent, well-conditioned observations needed for accurate recovery. Under the assumption that the action profiles that are never played with positive probability have payoffs strictly worse than the lowest payoff of any action profile played with non-zero probability, one can solve the optimization problem on the restricted set of action profiles that are observed in at least one equilibrium, and set the payoffs of the remaining action profiles to be lower than the lowest payoff of the recovered subgame, without affecting the equilibrium structure of the game. While the recovered game may not be the unique good explanation of the observations when looking at the full payoff matrix, it is unique with respect to the subgame of non-trivial actions when one has access to sufficiently many linearly independent, well-conditioned equilibrium observations.
\end{remark}

%% file: nopayoff_degeneracy.tex
\section{Finding consistent games without additional information}\label{sec: no_payoff_info}

This section focuses on the \textit{no payoff information} variant of the observation model given in Section~\ref{sec: model}. Recall that in this setting, the observer only observes what equilibrium $e^k$ is played for each perturbed game $G^k$. In this section, we note that in the absence of additional information, the consistent region contains a continuum of trivial and nearly trivial games that may not be of interest to an observer. Hence, we provide a framework that allows the observer to avoid recovering trivial games by controlling the degree of ``degeneracy'' (i.e., closeness to a trivial game) of the games he considers. Further, we characterize how much the size of the consistent set shrinks as a function of the minimum level of degeneracy of the games the observer is interested in.

\subsection{Finding non-degenerate games}

In this section, we separate the programs solved for players $1$ and $2$ and focus on the optimization problem that recovers the payoffs of player $1$ (by symmetry, all results can be applied to the optimization program that recovers the payoffs player $2$); we drop the player indices for notational simplicity. Since no payoff information is given, throughout this section, we assume w.l.o.g that the games are normalized to have all payoffs between $0$ and $1$. As mentioned in Example~\ref{ex: trivial_explanation}, the all-constant game $G=G^1=...=G^l$  gives an optimal solution to our optimization problem, as such a game is compatible with all equilibrium observations and has an objective value $d(G^1,...,G^l|G)=0$. It is therefore the case that solving our optimization problem might output a degenerate game, so in this section, we provide a framework that allows us to control the degree of degeneracy of the game we recover and to avoid trivial, all-constant games. To do so, we require some of the equilibria of the games to be ``strict,'' in the sense that 
\begin{align*}
\sum\limits_{j=1}^d G(i,j) x_{ij} \geq \sum\limits_{j=1}^d G(i',j) x_{ij}+\varepsilon_{ii'} \; \forall i,i'
\end{align*}
with $\varepsilon_{ii'} \geq 0$ and with the condition that at least one of the $\varepsilon_{ii'}$ is non-zero. All-constant games do not have strict equilibria, thus this avoids such games. Note that such a technique only affect the payoffs of pure strategies that are played with positive probability, and does not accord any importance to strategies that are never played. Let us now consider the new problem:
\begin{equation*}
\begin{array}{ll@{}ll}
\min\limits_{G^k,G}  & \displaystyle \displaystyle d(G^1,...,G^l|G) &\\
\text{s.t.}& e^k\text{ is a ``strict'' equilibrium of }G^k , \; \forall k\\
  		&0 \leq G(i,j) \leq 1, \; \forall (i,j)
\end{array}
\end{equation*}
which can be rewritten as

\begin{equation*}
\begin{array}{ll@{}ll}
\min\limits_{G^k,G}  & \displaystyle d(G^1,...,G^l|G) &\\
\text{s.t.}& \sum\limits_{j=1}^d G^k(i,j) e^k_{ij} = \sum\limits_{j=1}^d G^k(i',j) e^k_{ij}+\varepsilon^k_{ii'} \; \forall (i,i'), \forall k\\
  		&0 \leq G(i,j) \leq 1, \; \forall (i,j)
\end{array}
\end{equation*}

We introduce a positive parameter $\varepsilon$ that controls the level of non-degeneracy of the game and let the optimization program decide how to split $\varepsilon$ among the $\varepsilon^k_{ii'}$'s in a way that minimizes the objective. The optimization program can now be written as
\begin{equation*}
\begin{array}{ll@{}ll}
P(\varepsilon)= & \min\limits_{G^k,G} \; & \displaystyle  d(G^1,...,G^l|G)  &\\
& \text{s.t.}  &\sum\limits_{j=1}^d G^k(i,j) e^k_{ij} = \sum\limits_{j=1}^d G^k(i',j) e^k_{ij}+\varepsilon^k_{ii'} \; \forall (i,i'), \forall k\\
		&&\sum\limits_{k=1}^l \sum\limits_{i,i'}  \varepsilon^k_{ii'} = \varepsilon\\
		&&0 \leq G \leq 1\\
		&& \varepsilon^k_{ii'} \geq 0 \; \forall (i,i'), \forall k
\end{array}
\end{equation*}

For all $i,i' \in \cA_1$ such that $i \neq i'$ and $k \in [l]$, we introduce vectors $\tilde{e}^{k}_{ii'}$ whose entries are defined as follows:
\begin{align*}
\tilde{e}^{k}_{ii'}(h,j)=
\begin{cases} 
      -e^k(i,j) & \text{if $h=i$} \\
      e^k(i,j) &  \text{if $h=i'$} \\
      0 & \text{if $h \neq i,i'$} 
   \end{cases}
\end{align*}

This allows us to rewrite the optimization program under the following form:
\begin{equation}\label{primal_program}
\begin{array}{ll@{}ll}
P(\varepsilon)= & \min\limits_{G^k,G} \; & \displaystyle d(G^1,...,G^l|G) &\\
& \text{s.t.} &\sum\limits_{k,i,i'}  \tilde{e}^{k \; \prime}_{ii'}G^k = - \varepsilon\\
		&& \tilde{e}^{k \; \prime}_{ii'}G^k \leq 0 \; \forall (i,i'), \forall k\\
		&&0 \leq G \leq 1
\end{array}
\end{equation}

This optimization problem is, depending on the chosen metric, either a linear or quadratic optimization program with a tractable number of constraints, and can therefore be efficiently solved.

\subsection{A duality framework}

In this section, we give a duality framework under distance $d_2$ that offers insight into the solutions to the optimization program. Throughout the section, we let $D(\varepsilon)$ be the dual of Program~\eqref{primal_program}.

\subsubsection{Sufficient conditions for strong duality}

\begin{claim}
If there exist $G^1,...,G^l$ such that
$$ \tilde{e}^{k \; \prime}_{ii'}G^k < 0 \; \forall (i,i') \in cA_1, \forall k \in [l]  \text{ s.t. } \tilde{e}^k_{ii'} \neq 0,$$
then strong duality holds and $P(\varepsilon)=D(\varepsilon)$.
\end{claim}

\begin{proof}
Slater's condition holds iff there exists a solution $G,G^1,...,G^l$ such that 
\begin{align*}
&\sum\limits_{k,i,i'}  \tilde{e}^{k \; \prime}_{ii'}G^k = - \varepsilon\\
& \tilde{e}^{k \; \prime}_{ii'}G^k < 0 \; \forall (i,i'), \forall k  \text{ s.t. } \tilde{e}^k_{ii'} \neq 0.
\end{align*}
It is enough to find $G^1,...,G^l$ such that
$$\tilde{e}^{k \; \prime}_{ii'}G^k < 0 \; \forall (i,i'), \forall k \text{ s.t. } \tilde{e}^k_{ii'} \neq 0$$
as we can then renormalize the $G^k$'s such that $\sum\limits_{k=1}^l \sum\limits_{i,i'}  \varepsilon^k_{ii'} = \varepsilon$.
\end{proof}

Note that the previous sufficient condition is not necessarily tractable to check. We give a stronger sufficient condition such that for any fixed $k$, $\tilde{e}^{k \; \prime}_{ii'}G^k < 0 \; \forall (i,i')$ has a solution:

\begin{lemma}\label{duality_condition}
\label{conditions_Slater}
Let $k \in [l]$. Let $e^k(i,:)=(e^k(i,1),...,(e^k(i,m_2))$ $\forall i \in \cA_1$. If the non-null $e^k(1,:),...,e^k(m_1,:)$ are linearly independent, then the non-null $\tilde{e}^{k }_{ii'}$'s are linearly independent. In particular, there exists $G^k$ such that 
\begin{equation*}
\tilde{e}^{k \; \prime}_{ii'}G^k<0, \; \forall i,i' \in \cA_1.
\end{equation*}
If this holds for all $k \in [l]$, then $P(\varepsilon)=D(\varepsilon)$.
\end{lemma}

\begin{proof}
Let $\alpha(h,h')$'s be such that $\sum\limits_{h,h'} \alpha(h,h') \tilde{e}^{k}_{hh'}=0$, and so $\sum\limits_{h,h'} \alpha(h,h') \tilde{e}^{k}_{hh'}(i,j)=0$ $\forall (i,j)$. Recall that for a fixed $(i,j)$, $\tilde{e}^{k}_{h,h'(i,j)} \neq 0$ only if $h=i$ or $h'=i$, but not both at the same time. Therefore, 
$$\sum\limits_{h,h'} \alpha(h,h') \tilde{e}^{k}_{hh'}(i,j)= \sum\limits_{h' \neq i} \alpha(i,h') \tilde{e}^{k}_{i,h'}(i,j)+ \sum\limits_{h \neq i} \alpha(h,i) \tilde{e}^{k}_{h,i'}(i,j).$$
As $ \tilde{e}^{k}_{i,h'}(i,j)=-e^k(i,j)$ and $\tilde{e}^{k}_{h,i}(i,j)=e^k(h,j)$, we have for all $(i,j)$ that
$$
-e^k(i,j) \sum\limits_{h' \neq i} \alpha(i,h')+\sum\limits_{h \neq i} \alpha(h,i) e^k(h,j)= \sum\limits_{h,h'} \alpha(h,h') \tilde{e}^{k}_{hh'}(i,j)= 0.
$$
Since this holds for all values of $j$, it immediately follows that for all $i$,
$$
-e^k(i,:) \sum\limits_{h' \neq i} \alpha(i,h')+\sum\limits_{h \neq i} \alpha(h,i) e^k(h,:)=0.
$$

Take any $i,i'$ such that $e^k_{ii'} \neq 0$. Then $e^k(i,:) \neq 0$ and $e^k(i',:) \neq 0$. By the previous equation, we have
\begin{align*}
-e^k(i',:) &\sum\limits_{h' \neq i'} \alpha(i',h')+ \alpha(i,i') e^k(i,:) + \sum\limits_{h \neq i,i'} \alpha(h,i') e^k(h,:)
\\ &=-e^k(i',:) \sum\limits_{h' \neq i'} \alpha(i',h')+\sum\limits_{h \neq i'} \alpha(h,i') e^k(h,:)\\
&=0
\end{align*}
and by the linear independence assumption, we necessarily have $\alpha(i,i')=0$. Therefore, the $\tilde{e}^{k}_{ii'} \neq 0$'s are linearly independent, completing the proof.
\end{proof}

Note that in the worst case, we want $m_1 \leq m_2$, as there can be up to $m_1$  non-null $e^k(i,:)$ of size $m_2$, and by symmetry, we want $m_2 \leq m_1$ for the program that recovers the payoffs of player $2$. For the remainder of this section, we require $m_1=m_2$, which can be obtained by adding dummy actions to $\cA_p$ of player $p$ with the least available actions. When the condition does not hold in such a setting, it can be obtained through small perturbations of the equilibrium observations.


\subsubsection{Dual program}
The dual of program~\eqref{primal_program} is given by:

\begin{theorem}
The dual of optimization problem~\ref{primal_program} is given by:
\begin{equation}\label{dual_program}
\begin{array}{ll@{}ll}
D(\varepsilon)=\max\limits_{\mu^k_{ii'},\lambda_0,\lambda_1}  &-\frac{1}{4}\sum_{k=1}^l (\sum\limits_{i,i'} \mu^k_{ii'} \tilde{e}^{k}_{ii'})'  (\sum\limits_{i,i'} \mu^k_{ii'} \tilde{e}^{k}_{ii'})-\mathbbm{1}' \lambda_1 - \mu \varepsilon\\
\text{s.t.}	&\lambda_1-\lambda_0+ \sum\limits_{k,i,i'} \mu^k_{ii'} \tilde{e}^{k}_{ii'} = 0
\\		&\mu+\mu^k_{ii'} \geq 0
\\		&\lambda_0,\lambda_1 \geq 0                  	
\end{array}
\end{equation}
The KKT conditions imply that if $(G^{1*},...,G^{l*},G^*)$ is a primal optimal solution and $(\lambda_0^*,\lambda_1^*,\mu^*,\mu^{k \; *}_{ii'})$ is a dual optimal solution, then
\begin{equation}\label{KKT}
\begin{array}{ll@{}ll}
\forall k, G^{k*}=A-\frac{1}{2}  (\frac{\lambda_1^*-\lambda_0^*}{l} +\sum\limits_{i,i'} \mu^{k \; *}_{ii'} \tilde{e}^{k}_{ii'})\\
G^*=A-\frac{1}{2l} (\lambda_1^*-\lambda_0^*)
\end{array}
\end{equation}
for some matrix  $A \in \mathbb{R}^{l \times l}$

\end{theorem}

\begin{proof}
See Appendix~\ref{proof_dual_program}.
\end{proof}

This duality framework will allow us to obtain bounds on the trade-off between degeneracy and accuracy in the next subsection. 

\subsection{Trade-off between degeneracy and objective value}

\begin{definition}
We define the \textit{degeneracy threshold} $\varepsilon^*$ of a set of observations as
$$\varepsilon^*=\sup \{\varepsilon \text{ s.t. } P(\varepsilon)=0\}.$$
\end{definition}

\begin{claim}
The degeneracy threshold is given by
\begin{equation}\label{threshold_epsilon}
\begin{array}{ll@{}ll}
\varepsilon^*= & - \min\limits_{G,\varepsilon^k_{ii'}} \; & \displaystyle \sum\limits_{k,i,i'} \tilde{e}^{k \; \prime}_{ii'}G\\
& \text{s.t.} & \tilde{e}^{k \; \prime}_{ii'}G \leq 0 \; \forall (i,i'), \forall k\\
		&&0 \leq G \leq 1
\end{array}
\end{equation}
\end{claim}

The claim gives a tractable linear program to solve for the degeneracy threshold.

\begin{proof}
Remark that $\varepsilon^*$ solves
\begin{align*}
\begin{array}{ll@{}ll}
\varepsilon^*= & \max\limits_{G^k,G,\varepsilon^k_{ii'}} \; & \displaystyle \varepsilon\\
& \text{s.t.}  &\tilde{e}^{k \; \prime}_{ii'}G^k+\varepsilon^k_{ii'} = 0 \; \forall (i,i'), \forall k\\
		&&\sum\limits_{k=1}^l (G^k-G)'(G^k-G)=0\\
		&&\sum\limits_{k,i,i'}  \varepsilon^k_{ii'} = \varepsilon\\
		&& \varepsilon^k_{ii'} \geq 0 \; \forall (i,i'), \forall k\\
		&&0 \leq G \leq 1
\end{array}
\end{align*}

From the fact that $\sum\limits_{k=1}^l (G^k-G)'(G^k-G)=0$ implies $G^1=...=G^l=G$, we have

\begin{equation*}
\begin{array}{ll@{}ll}
\varepsilon^*= & \max\limits_{G,\varepsilon^k_{ii'}} \; & \displaystyle \sum\limits_{k,i,i'}  \varepsilon^k_{ii'}\\
& \text{s.t.}  &\tilde{e}^{k \; \prime}_{ii'}G+\varepsilon^k_{ii'} = 0 \; \forall (i,i'), \forall k\\
		&& \varepsilon^k_{ii'} \geq 0 \; \forall (i,i'), \forall k\\
		&&0 \leq G \leq 1
\end{array}
\end{equation*}
The result follows immediately.
\end{proof}

\begin{claim}
 $\varepsilon^*$ is finite, $\forall \varepsilon \leq \varepsilon^*$, $P(\varepsilon) = 0$, and $\forall \varepsilon > \varepsilon^*$, $P(\varepsilon) > 0$. 
\end{claim}

\begin{proof}
The proof follows immediately from Claim~\ref{lower_bound}. $P(\varepsilon^*)=0$ comes from the fact that the feasible set of Program~\eqref{threshold_epsilon} is bounded: indeed, for any point in its feasible set, $0 \leq G \leq 1$ and $\varepsilon^k_{ii'}=-\tilde{e}^{k \; \prime}_{ii'}G$, forcing the $\varepsilon^k_{ii'}$ to also be bounded. Thus, $\varepsilon^*$ is a solution of a linear program on a bounded polytope and is therefore finite, and attained at an extreme point of this polytope. 
\end{proof}

Note that if we solve optimization Program~\eqref{threshold_epsilon} and find that $\varepsilon^*$ is large, then it is possible to find a large value of $\varepsilon$ such that $P(\varepsilon)=0$, and we can therefore recover a non-degenerate game that has all of the observed equilibria, i.e., we recover a game that has equilibrium properties similar in some sense to those of the true, underlying game. For smaller values of $\varepsilon^*$, we refer to the following statement:

\begin{theorem}\label{degeneracy_accuracy_trade_off}
For every $\varepsilon_0 > \varepsilon^*$, and for all $\varepsilon \geq \varepsilon_0$, we have $f(\varepsilon) \leq P(\varepsilon) \leq g(\varepsilon)$ where $f$ and $g$ are given by
\begin{align}
 &f(\varepsilon)=P(\varepsilon_0) \frac{\varepsilon^2}{\varepsilon_0^2}\\
 &g(\varepsilon)=\Big( (\sqrt{P(\varepsilon_0)}+\frac{\sqrt{l} m}{2})\frac{\varepsilon}{\varepsilon_0}-\frac{\sqrt{l} m}{2} \Big) ^2
\end{align}
\end{theorem}

\begin{proof}
See Appendix~\ref{proof_degeneracy_accuracy_trade_off}.
\end{proof}

An observer that sets a degeneracy parameter of $\varepsilon$ restricts himself to a set of games that must have an empty intersection with $S_d(f(\varepsilon))$ and a non-empty intersection with the set $S_d(g(\varepsilon))$.

%% file: extensions.tex
\section{Extensions}\label{sec: extensions}

In this section, we show extensions of our framework---first to succinct games with many players, and second to some games with infinite action spaces.

\subsection{Linear succinct games (as per~\cite{KS2015})}\label{sec: succinct}

In general, computational tractability cannot be achieved as the number of players increases. A reason for this is that in the general case, an intractable, exponential number (in the number of players) of variables need be used to represent the game and its equilibria: in a game with $n$ players and $m$ actions per player, there are $m^n$ pure action profiles, hence $m^n$ variables are needed simply to represent the payoff matrices and the equilibria of the recovered games. 

However, if the game and the observed equilibria have a compact representation, the equilibrium constraints can be written down using a tractable number of variables, and our framework provides efficient algorithms to find an element in the consistent set, compute its diameter, and test for linear properties. Kuleshov and Schrijvers~\cite{KS2015} consider linear succinct games and show that if the structure of the succinct game is known and if we observe an equilibrium such that the ``equilibrium summation property holds'' (roughly, the exact expected utility of the players can be computed efficiently), then a game is consistent with the equilibrium observations if and only if a polynomial number of tractable, linear constraints are satisfied. Such constraints can easily be incorporated into our framework. (See Property $1$ and Lemma $1$ of~\cite{KS2015} for more details.)

\subsection{Cournot competition and infinite action space}\label{sec: cournot}

In the general case, our framework cannot directly deal with games with infinite action spaces in a tractable way: to write down an equilibrium constraint, one needs a constraint for each of the infinite number of possible deviations. In this section, we show that, nevertheless, for some games with infinite action spaces, only a finite, tractable number of constraints is needed to characterize the equilibria; we show how to adapt our framework to such games. For illustration, we focus on the Cournot competition game with continuous spaces of production levels. 

Consider a Cournot competition with $n$ players selling the same good. Each player $i$ chooses a production level $q_i \geq 0$, and sells all produced goods at price $P(q_1,...,q_n)$ common to all players, and each player $i$ incurs a production cost $c_i(q_i)$ to produce $q_i$ units of the good; we write G=$(P,c)$ where $c=(c_1,...,c_n)$. We assume that $P$ is concave in each $q_i$.

We assume the observer knows the function $P$ and wants to recover the costs $c_i$ of the players, where the underlying costs change slightly with each observation. Formally, consider that we have $l$ perturbed games such that in every pertubed game $k$, the players play a Cournot competition with the same, commonly known price function $P$ but perturbed cost functions $c^k_i$ for each player $i$, known to be convex. We obtain equilibrium observations $q^1,...,q^l$, where $q^k_i$ is the equilibrium production level of player $i$ in perturbed game $k$ and $q^k=(q^k_1,...,q^k_n)$, that is, $G^k=(P,c^k)$. 
%

Suppose the following hold:
\begin{itemize}
\item The observer knows the costs belong to the space of polynomials of any chosen fixed degree $d \geq 1$; i.e., the observer parameterizes the underlying and perturbed cost functions in the following way:
\begin{equation}\label{eq: poly_cost}
c_i(q_i)=\sum\limits_{ex=1}^d a_i(ex) q_i^{ex}
\end{equation}
where the $a_i(ex)$'s are now the variables the observer want to recover.
\item $d(c^1,...,c^l|c) \leq \delta$ can be written as a tractable number of semidefinite constraints on the $a_i^k$'s and $a_i$'s (this includes, but is not limited to, the $d_2$ and $d_{\infty}$ distances).
\item $P(q)$ and $\frac{\partial P}{\partial q_i}$ can be computed efficiently for $i$, given $q$. 
\end{itemize}
Then $S_d(\delta)$ has an efficiently computable and tractable representation (as a function of $n$, $l$ and $d$) as the intersection of SDP constraints. This means in particular that optimizing a linear function over $S_d(\delta)$ can be cast as a tractable semidefinite program, for which efficient solvers are known -- see~\cite{boyd2004convex}; one such solver, that we use in the simulations of Section~\ref{sec:simulations}, is CVX~\cite{CVX}. This enables us to efficiently recover a game in the consistent set, efficiently compute its diameter, and efficiently test for linear properties (by simply adding linear and thus SDP constraints, as needed).

To obtain such a tractable characterization, we only need to note that i) the equilibrium constraints can be rewritten as a tractable number of tractable linear constraints, and ii) convexity constraints on polynomials can be classically cast as tractable SDP constraints. This is the object of Appendix~\ref{sec: cournot_eq} and~\ref{sec: cournot_convex}.

%% file: simulations.tex
\section{Simulations for entry games and Cournot competition}\label{sec:simulations}


In this section, we run simulations for two concrete settings to illustrate the power of our approach. We first (Section~\ref{sim: entry}) illustrate how our framework performs on a simple entry game. We then (Section~\ref{sim: cournot}) show that it is able to handle much larger games.

\subsection{2-player entry game}\label{sim: entry}

We first consider an {\em entry game}, in which each of two players (think of them as companies deciding whether to open a store in a new location) has two actions available to him (enter the market; don't enter the market). Entry games are common in the literature, as seen in~\cite{ABJ2004,CT2009,BMM2011}, and, because of their simplicity, allow us to cleanly visualize the consistent region. 


Each player $p$ has two actions: $\cA_p=\{0,1\}$; $a_p=0$ if player $p$ does not enter the market, $a_p=1$ if he does. The utility of a player is given by $G_p(a_p,a_{-p})=a_p ((1-a_{-p}) \gamma_p + a_{-p} \theta_p)$ for some parameters $\gamma_p \geq 0$ and $\theta_p \leq \gamma_p$, similarly to~\cite{T2003}: if player $p$ does not enter the market, his utility is zero; if he enters the game but the other player does not, $p$ has a monopoly on the market and gets non-negative utility; finally, if both players enter the game, they compete with each other and get less utility than if they had a monopoly. 

In our simulations, we fix values for the parameters $(\gamma_p,\theta_p)$ and generate the perturbed games as follows: 
\begin{itemize}
\item In the partial payoff information setting, we add independent Gaussian noise with mean $0$ and standard deviation $\sigma$ to $G_p(a_p=1,a_{-p})$  (we vary the value of $\sigma$) to obtain the perturbed games $G^1,...,G^l$. 
\item In the payoff shifter information setting, we sample the payoff shifters $\beta^1,...,\beta^l$ such that for all $k \in [l]$, for all players $p$, $\beta^k_p(a_p=1,a_{-p})$ follows a normal distribution of mean $0$ and standard deviation $\sigma_{s}$. We then add Gaussian noise with mean $0$ and standard deviation $\sigma$ to $G_p(a_p=1,a_{-p})$ to obtain the perturbed games $G^1,...,G^l$. 
\end{itemize}
In both observation models, paralleling the setting of~\cite{BMM2011}, no observed payoff shifter nor unknown noise is added to the payoff of action $a_p=0$ for player $p$; action $a_p=0$ is always assumed to yield payoff $0$ for player $p$, independently of $a_{-p}$. In order to generate the equilibrium observations, once the perturbed games are generated, we find the set of equilibria of each of the $G^k$, and sample a point $e^k$ in said set. In the payoff information case, we also compute $v^k=e^{k \; \prime} G^k$.

In order to parallel the setting of Beresteanu et al.~\cite{BMM2011}, we assume the observer knows the form of the utility function, i.e., that $G_p(0,0)=0$ and $G_p(a_p=0,a_{-p}=1)=0$, and that he aims to recover the values of $\gamma_p$ and $\theta_p$. Thus, we add linear constraints $G_p(0,0)=0$ and $G_p(a_p=0,a_{-p}=1)=0$ in the optimization programs that we solve (see Program~\eqref{optim_unseparated}) in the payoff shifter information and partial payoff information settings. Furthermore, we assume as in~\cite{BMM2011} that the observer knows that perturbations are only added to $\gamma$ and $\theta$, and therefore we add linear constraints $G_p^k(0,0)=0$ and $G_p^k(a_p=0,a_{-p}=1)=0$ for all $k \in [l]$ to the optimization problems for player $p$ in each of the observation models. All optimization problems are solved in Matlab, using CVX (see~\cite{CVX}). 

Our model for entry-games is similar to the ones presented in~\cite{T2003} and used in simulations in~\cite{BMM2011}, so as to facilitate informal comparisons of the simulation results of both papers; in particular, the parametrization of the utility functions of the players in our simulations is inspired by~\cite{BMM2011}, and noise is generated and added in a similar fashion. However, while we attempt to parallel the simulations run by Beresteanu et al.~\cite{BMM2011}, it is important to note that this is not an apples-to-apples comparison, because of key differences in the setting. In particular, our observation models (seeing full equilibria) and the information available to the observer (no distributional assumptions) are different from those in~\cite{BMM2011}.

\subsubsection{Consistent regions for Player 1}

We fix $l=500$, $\gamma=5$, $\theta=-10$ in all simulations, and vary the values of $\sigma$ and $\sigma_s$. Because the observations are generated by adding i.i.d Gaussian noise with mean $0$ and variance $\sigma^2$ to the two payoffs for entry of each player, if $G$ is the underlying game and $G^1,...,G^l$ are its perturbations, 
$$
\frac{1}{\sigma^2} d_2(G^1,...,G^l|G) \text{ (resp. } \frac{1}{\sigma^2} d_2(G^1-\beta^1,...,G^l-\beta^l|G) \text{)}
$$
follows a Chi-square distribution with $4l$ degrees of freedom in the partial payoff information case (resp. in the payoff shifter case). We choose $\delta$ such that $P(d_2(G^1,...,G^l|G)  \leq \delta) \approx 0.99$, and suppose the observer sees said value of $\delta$. While the observer does not have access to the distribution of perturbations, it is extremely likely he will observe a magnitude of perturbations equal to or less than $\delta$, and we can use $\delta$ as a high-probability upper bound on the information on the perturbations accessible to the observer.

\begin{figure}[!h]
\begin{subfigure}{0.24\textwidth}
  \centering
  \includegraphics[width=1.1\linewidth]{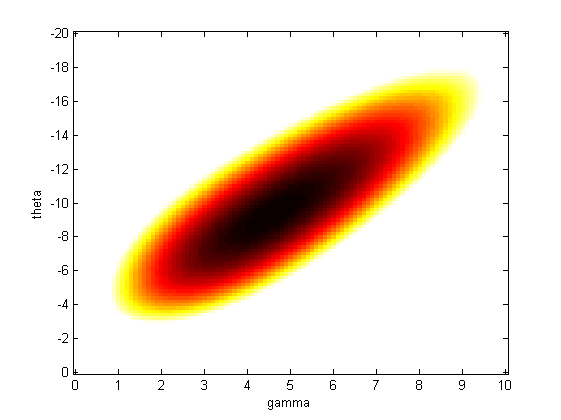}
  \caption{$\sigma=0.5$, \\~$\sigma_s=2.5$}
\end{subfigure}
\begin{subfigure}{0.24\textwidth}
  \centering
  \includegraphics[width=1.1\linewidth]{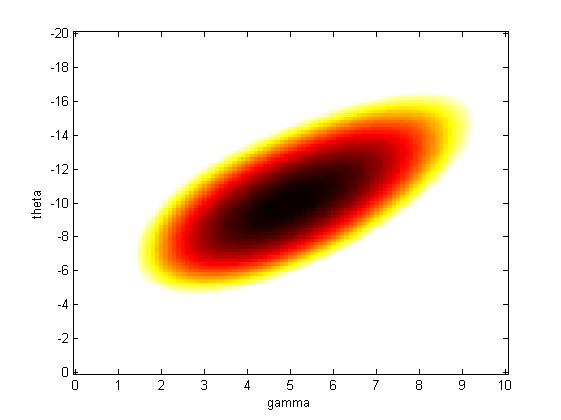}
  \caption{$\sigma=0.5$, \\~$\sigma_s=5$}
\end{subfigure}
\begin{subfigure}{0.24\textwidth}
  \centering
  \includegraphics[width=1.1\linewidth]{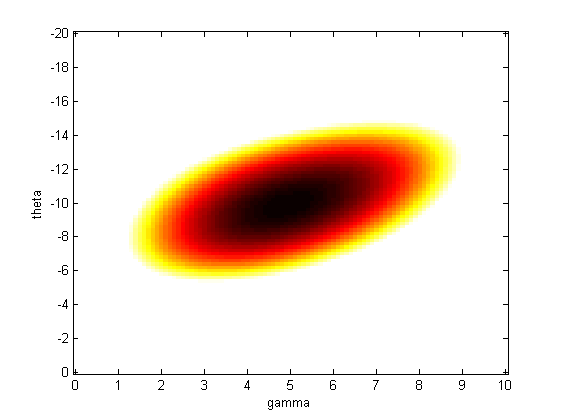}
  \caption{$\sigma=0.5$, \\~$\sigma_s=10$}
\end{subfigure}
\begin{subfigure}{0.24\textwidth}
  \centering
  \includegraphics[width=1.1\linewidth]{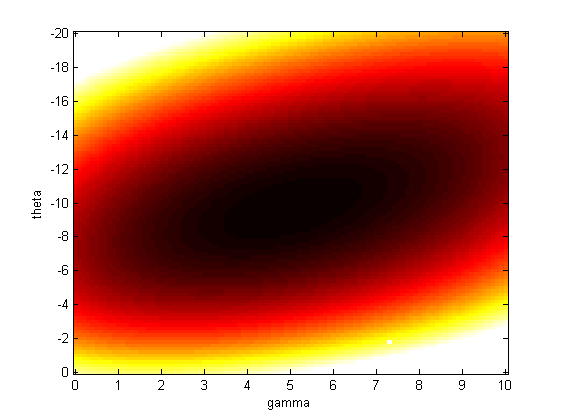}
  \caption{$\sigma=1.5$, \\~$\sigma_s=10$}
\end{subfigure}
\caption{Plots of the consistent region for different values of $\sigma,\sigma_s$ in the payoff shifter information observation model}
\label{fig: shift_results}
\end{figure}
In all plots, the colored region in the plots is the projection over the space $(\gamma_1,\theta_1)$ for player $1$ of the set of parameters $(\gamma_1,\theta_1,\gamma_2,\theta_2)$ that are in the $\delta$-consistent region. The darker the region, the smaller the objective value of the best explanation for the corresponding values of $\gamma$ and $\theta$. The black, center of the region represents the value of $(\gamma,\theta)$ that minimizes $d_2(G^1,...,G^l|G)$.

Figure~\ref{fig: shift_results} shows the evolution of the consistent region when varying $\sigma$ and $\sigma_s$ in the payoff shifter information setting. The smaller the standard deviation $\sigma$ of the unknown noise, the tighter the consistent region. On the other hand, reasonably increasing the value of $\sigma_s$ can be beneficial, at least when it comes to centering the consistent region on the true values of the parameters: this comes from the fact that when the game is sufficiently perturbed, new equilibria arise and new, informative behavior is observed, while not adding significant uncertainty to the payoffs of the game.

Figure~\ref{fig: payoff_results} shows the evolution of the consistent region when varying $\sigma$. The larger the value of $\sigma$, the larger the consistent region, and the further away its center is from the underlying, true value of the parameters. 
\begin{figure}[!h]
\begin{subfigure}{0.24\textwidth}
  \centering
  \includegraphics[width=1.1\linewidth]{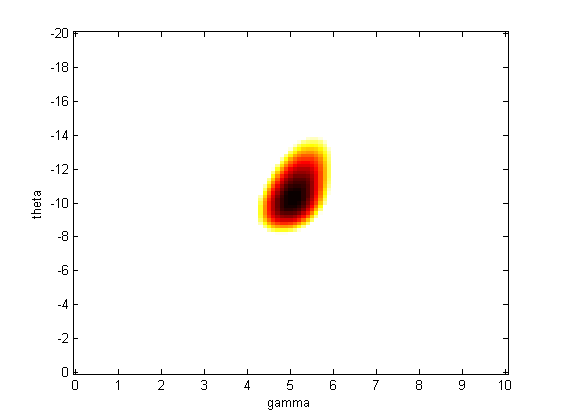}
  \caption{$\sigma=0.5$}
\end{subfigure}
\begin{subfigure}{0.24\textwidth}
  \centering
  \includegraphics[width=1.1\linewidth]{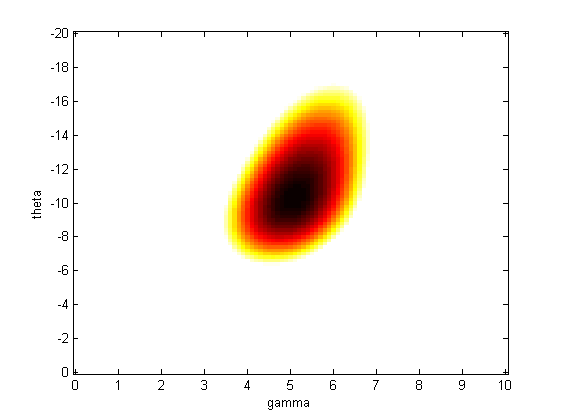}
  \caption{$\sigma=1.0$}
\end{subfigure}
\begin{subfigure}{0.24\textwidth}
  \centering
  \includegraphics[width=1.1\linewidth]{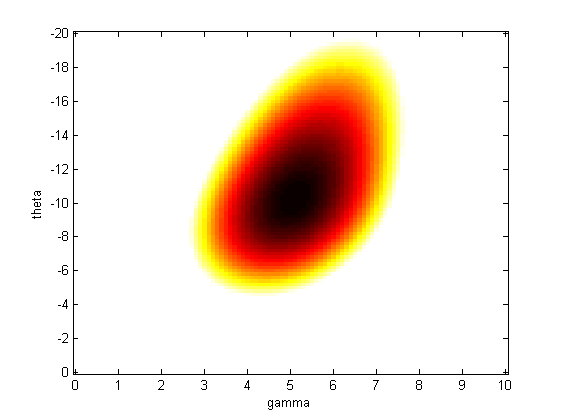}
  \caption{$\sigma=1.5$}
\end{subfigure}
\begin{subfigure}{0.24\textwidth}
  \centering
  \includegraphics[width=1.1\linewidth]{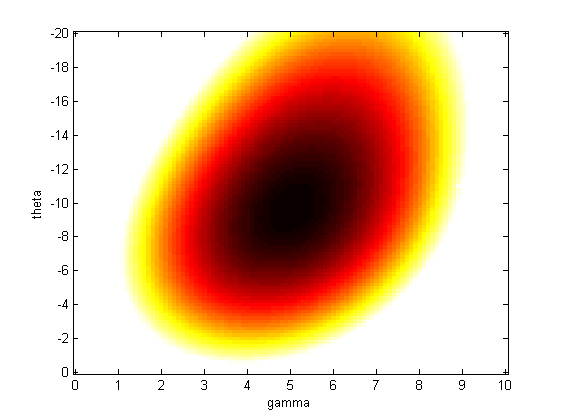}
  \caption{$\sigma=2.5$}
\end{subfigure}

\caption{Plots of the consistent region for different values of $\sigma$ in the partial payoff information observation model}
\label{fig: payoff_results}
\end{figure}

\subsubsection{Testing for linear properties} We also illustrate via simulation how our framework can test the ability of linear properties to explain observed behavior. In particular, here we test whether a set of observations is likely to be explained by a zero-sum game. We consider entry games as defined in the previous section, and assume the observer wants to test whether observations were generated by a game that is approximately zero-sum, without any information on the parametric form of the game (the observer does not know the game is an entry game).

\begin{figure}
\begin{subfigure}{0.45\textwidth}
  \centering
  \includegraphics[width=1.1\linewidth]{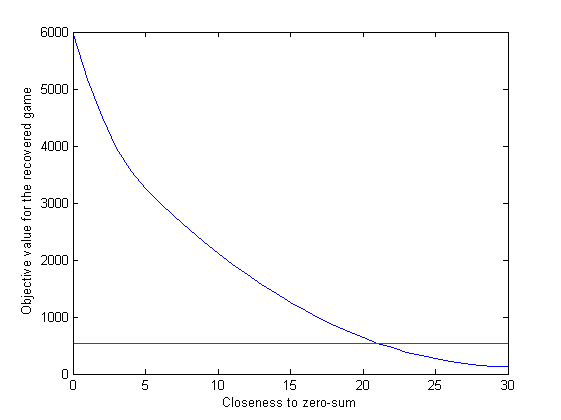}
  \caption{Payoff information setting with $\sigma=0.5$}
\end{subfigure}
\begin{subfigure}{0.45\textwidth}
  \centering
  \includegraphics[width=1.1\linewidth]{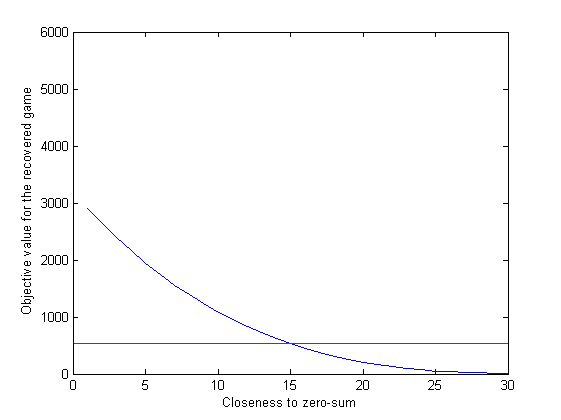}
  \caption{Payoff shifter information setting with $\sigma=0.5$, $\sigma_s=10$}
\end{subfigure}
\caption{Testing for zero-sum with respect to the $1$-norm}
\label{fig:zerosum}
\end{figure}

Formally, we say a game $G=(G_1,G_2)$ is $\varepsilon$-zero-sum with respect to the $p$-norm if and only if $\Vert G_1+G_2 \Vert_{p} \leq \varepsilon$. Note that a game being $\varepsilon$-zero-sum is a linear property and therefore can be included in our framework. The smaller the value of $\varepsilon$, the more stringent the condition is and the closer $G$ must be to a zero-sum game. We use $l=500$, $\sigma=0.5$, $\sigma_s=10$ in all simulations.

As before, we pessimistically assume the observer sees $\delta$ such that 
\\$P(d_2(G^1,...,G^l|G) \leq \delta) \approx 0.99$, that is, $\delta=537.5$ for $l=500$, $\sigma=0.5$. Figure~\ref{fig:zerosum} shows for which values of $\varepsilon$ one can recover a $\varepsilon$-zero-sum game with objective value less than $537.5$ that explains the observations for different values of $\gamma$ and $\theta$. Values of $\varepsilon$ to the left of the intersection between the red and the blue line are impossible, while values to the right of this intersection indicate there is a $\varepsilon$-zero-sum game that explain the observations. In both cases, we see that no zero-sum game or game close to being zerosum is a good explanation for the observations; in the payoff information setting, no game less than $21$-zerosum explains the observations, while in the payoff shifter setting, no game less than $15$-zerosum explains the observations.

\subsection{Multiplayer Cournot competition}\label{sim: cournot}

 In this section, we run simulations on a Cournot competition with varying number of players. See Section~\ref{sec: cournot} for a discussion of how our framework can be modified to accommodate Cournot games with many players and an action set of infinite size for each player. All simulations were performed on a laptop with an Intel Core i7-4700MQ at 2.40GHz and $16$ GB RAM.

\subsubsection{Generating the games} Let $n$ be the number of players, and $q_i$ the production level of player $i$. We fix a parameter $\alpha=0.05$, and set the price function to be given by $P(q_1,...,q_n)=1-\alpha \sum\limits_{i=1}^n q_i$; the price function is known to the observer. We fix the form of the cost function to be linear; that is, the cost of producing $q_i$ of goods incurred by player $i$ is given by $c_i(q_i)=a_i q_i + b_i$. Without loss of generality, we set $b_i=0$: $b_i$ does not affect the maximization problem nor the first order condition solved by player $i$, and hence does not impact the chosen production level of the players. 

We generate $n_g=10$ underlying Cournot games with heterogeneous, linear cost functions $c_i(x)=a_i x$ as follows:
\begin{itemize}
\item We first set $\hat{a}_i=0.01$ for every player $i$.
\item We generate each of the $n_g$ games by adding i.i.d. truncated Gaussian noise $X_i$ with mean $0$ and standard deviation $0.01$ to the $\hat{a}_i$'s. I.e., $a_i=0.01+X_i$ where $X_i$ can be written as $X_i=max(Z_i,-0.01)$ and $Z_i$ is a non-truncated Gaussian with mean $0$, standard deviation $0.01$. This ensures the $a_i$'s are always non-negative, hence the  production costs are always non-decreasing. 
\end{itemize}
Note that the same $n_g$ games are used in all plots and simulations.

For each of the $n_g$ games, we then generate $l$ perturbed games by adding truncated Gaussian noise with standard deviation $\sigma=0.001$ to each of the $a_i$'s. As before, the noise is truncated to ensure non-negativity of the perturbed $a_i$'s. We then solve the first-order condition to obtain equilibrium observations and note that all obtained $q_i$'s are positive with extremely high probability.

\subsubsection{Observer's problem and simulation results}

We assume the observer wants to recover a cost that is polynomial of chosen degree $d \geq 1$. I.e., the observer parameterizes the cost functions in the following way: $c_i(x)=a_i(d) x^d + a_i(d-1) x^{d-1}+...+a_i(1) x$; we always assume that for every player $i$, every perturbed game $k$, $a_i(0)=0$ and $a_i^k(0)=0$ for simplicity (not producing anything costs the players nothing); we note that this without loss of generality, as a constant shift does not change the utility-maximizing strategies of the players. The observer also knows the perturbed cost function $c_i^k$'s of the perturbed games are convex. The program then solved by the observer is derived from the results of Section~\ref{sec: cournot}. Throughout this subsection, all results are averaged over the $n_g$ games originally generated: for each game, we measure the diameter of the consistent set, the time taken to recover a game within the set, and the time taken to compute the diameter, then average it over the $n_g$ games we consider.

\begin{figure}
\centering
\includegraphics[width=0.5\linewidth]{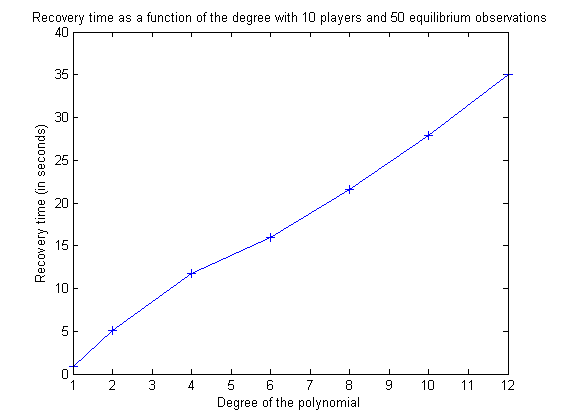}
 \caption{Average recovery time as a function of the degree with $10$ players, $10$ equilibrium observations}
\label{fig: degree_time}
\end{figure}

Figure~\ref{fig: degree_time} shows the time it takes the observer to recover a game within the consistent set as a function of the degree of the polynomial cost function considered, fixing the number of players to $10$ and the number of perturbed games/equilibrium observations to $50$ per underlying game. We see the recovery time is less than a minute, even when considering polynomials of degree $10$ or $12$, and that said time increases roughly linearly in the degree of the polynomial used for recovery. This allows for recovery within minutes for high-degree polynomials, even when using the minimal computing power of a personal laptop. 

Figure~\ref{fig: diameter_value} shows the value of the diameter of the consistent set when the observer assumes the cost function is linear. The diameter is plotted as a function of the number of observations, in the presence of a fixed number of players (10). The figure shows the diameter decreases quickly as the number of equilibrium observations. When only one equilibrium observation is available, the diameter is  $3.3\times10^{-3}$, which is $35 \%$ of the expected true cost $\hat{a}_i=0.01$; at $10$ equilibrium observations, it is $3.3 \times 10^{-4}$, i.e. only $3.5 \%$ of $\hat{a}_i$, and at $100$ equilibrium observations, it is $6.6 \times 10^{-5}$, i.e. $0.66 \%$ of $\hat{a}_i$. Hence, very few equilibrium observations are necessary to recover the underlying game accurately. 

\begin{figure}
\centering
\includegraphics[width=0.5\linewidth]{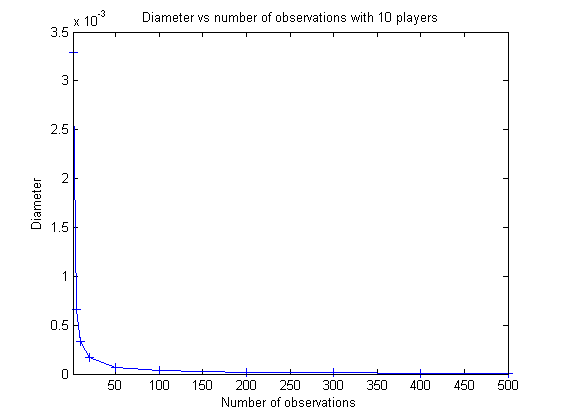}
\caption{Average diameter as a function of the number equilibrium observations}
\label{fig: diameter_value}
\end{figure}

\begin{figure}[H]
\begin{subfigure}{0.24\textwidth}
  \centering
  \includegraphics[width=1.0\linewidth]{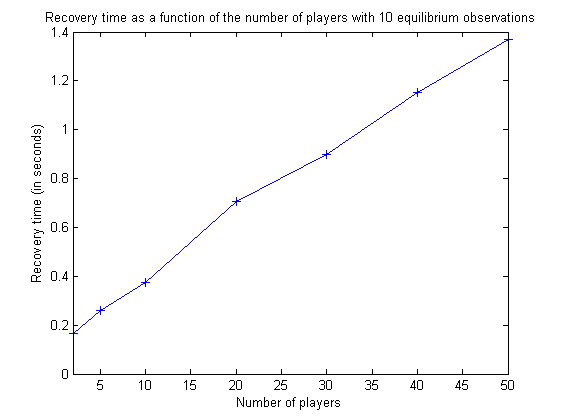}
  \caption{Recovery time\\vs players}
\label{fig: player_time_a}
\end{subfigure}
\begin{subfigure}{0.24\textwidth}
  \centering
  \includegraphics[width=1.0\linewidth]{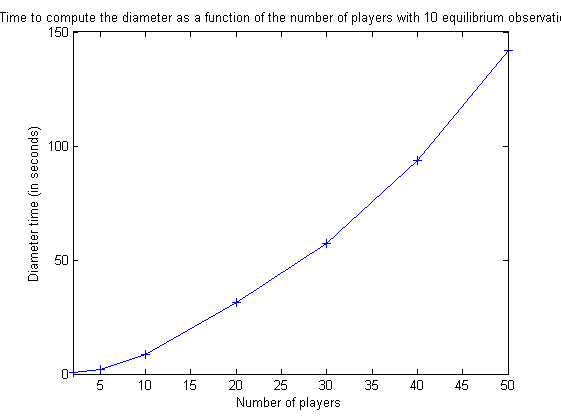}
  \caption{Diameter time \\vs players}
\label{fig: player_time_b}
\end{subfigure}
\begin{subfigure}{0.24\textwidth}
  \centering
  \includegraphics[width=1.0\linewidth]{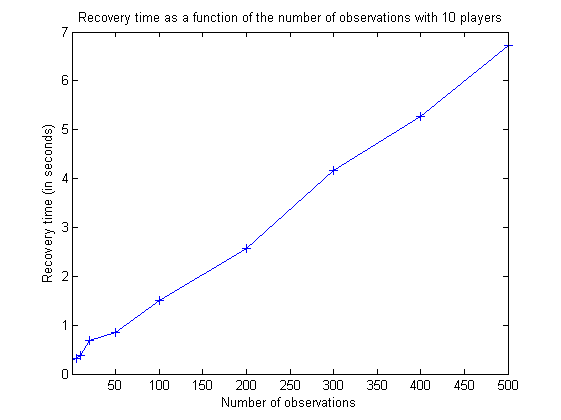}
  \caption{Recovery time\\vs equilibria}
\label{fig: eq_time_a}
\end{subfigure}
\begin{subfigure}{0.24\textwidth}
  \centering
  \includegraphics[width=1.0\linewidth]{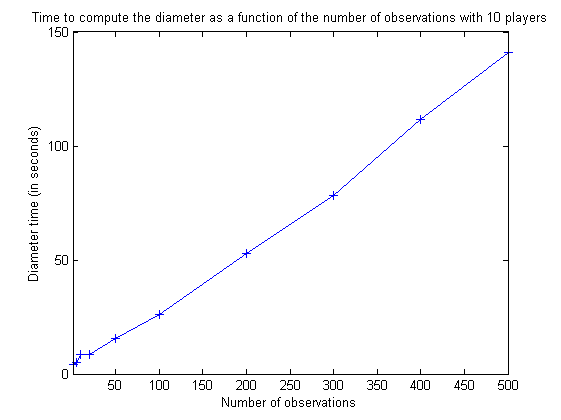}
  \caption{Diameter time\\vs equilibria}
\label{fig: eq_time_b}
\end{subfigure}
\caption{Average recovery/diameter time as a function of the number of observations/number of players, when recovering linear costs}
\end{figure}

Figures~\ref{fig: player_time_a} and~\ref{fig: player_time_b} show the time it takes the observer to recover a game within the consistent set and to compute its diameter as a function of the number of players, fixing the degree of the polynomial to $1$ and equilibrium observations to $50$ per underlying game. Similarly to before, we see that the recovery time increases roughly linearly as a function of the number of players, and it takes less than $1.5$ seconds on average to recover a consistent game, even with $50$ players. The diameter time increases roughly quadratically: this was expected, as both the number of programs to solve and the size of each program increase linearly in the number of players. Hence, while computing the diameter scales superlinearly, it remains computable within minutes, even with larger numbers of players, on a personal laptop. To the best of our knowledge, no previous approach to understanding the games consistent with observed behavior offers comparable scalability.

Finally, Figures~\ref{fig: eq_time_a} and~\ref{fig: eq_time_b} show the time it takes the observer to recover a game within the consistent set and to compute its diameter as a function of the number of equilibria, fixing the degree of the polynomial to $1$ and the number of players to $10$ per underlying game. Both the recovery time and the diameter time scale linearly with the number of equilibrium observations. Unlike before, while the size of each subprogram to solve to compute the diameter increases, the number of such programs is independent of the number of observations, allowing for extremely good scalability of our framework as function of the number of equilibrium observations. In practice, the number of players is fixed but the observer may see more observations over time, and our framework is able to deal with such an increasing number of observations.

%% file: appendix.tex
\section{Proof of performance of the algorithm}\label{sec: diameter_proof}
Consider the following optimization program:
\begin{equation*}
\begin{array}{ll@{}ll}
P_{\delta}= & \sup\limits_{\tilde{G},\hat{G},\gamma} \; & \displaystyle \gamma\\
& \text{s.t.} & \tilde{G} \in S_d(\delta)\\
		&&\hat{G} \in S_d(\delta)\\
		&&\max(\Vert \tilde{G}_1-\hat{G}_1 \Vert_{\infty},\Vert \tilde{G}_2-\hat{G}_2 \Vert_{\infty}) \geq \gamma\\

\end{array}
\end{equation*}
Clearly, $P_{\delta}=D(S_d(\delta))$, simply by noting that the program is a rewriting of Definition~\ref{def: diameter}. Now if $ P_{\delta} = \max\limits_{(i,j) \in m_1 m_2} \max (P_{\delta,1}(i,j),P_{\delta,2}(i,j))$, then we have shown that $\mathcal{A}(\delta)=D(S_d(\delta))$. This holds because:
\begin{itemize}
\item For every player $p$ and action profile $(i,j)$, if $\tilde{G}_p(i,j)-\hat{G}_p(i,j) \geq \gamma$ then $\max(\Vert \tilde{G}_1-\hat{G}_1 \Vert_{\infty},\Vert \tilde{G}_2-\hat{G}_2 \Vert_{\infty}) \geq \gamma$. Hence $P_{\delta} \geq \max\limits_{p,(i,j)} P_{\delta,p}(i,j)$.
\item If $\max(\Vert \tilde{G}_1-\hat{G}_1 \Vert_{\infty},\Vert \tilde{G}_2-\hat{G}_2 \Vert_{\infty}) \geq \gamma$, then there exists a player $p$ and a set of actions $(i,j)$ such that $\tilde{G}_p(i,j)-\hat{G}_p(i,j) \geq \gamma$ w.l.o.g. (remember $\tilde{G}$ and $\hat{G}$ play symmetric roles) and so one of the $P_{\delta,p}(i,j)$ must have objective value at least $\gamma$. This means that $P_{\delta} \leq \max\limits_{p,(i,j)} P_{\delta,p}(i,j)$.
\end{itemize}

\section{Proof of recovery lemma under infinite norm and payoff information}\label{infinite_norm_proof}

\begin{proof}
For simplicity of notation, we drop the indices $p$. We first remark that $(G,G^1,\ldots,G^l)$ is feasible for Program~\eqref{primal_program_payoff}; as $(\hat{G},\hat{G}^1,\ldots,\hat{G}^l)$ is optimal, it is necessarily the case that 

$$\max\limits_k\Vert \hat{G}-\hat{G}^k \Vert_{\infty}  \leq \max\limits_k  \Vert G-G^k \Vert_{\infty} \leq \delta.$$ 

Let us write $\Delta G=G-\hat{G}$. We know that for all $k$, $e^{k \; \prime} G^k=e^{k \; \prime} \hat{G}^k=v^k$, and thus $e^{k \; \prime} (G^k -\hat{G}^k)=0$. We can write 
\begin{align*}
E \Delta G &= (e_1'(G-\hat{G}) \; \ldots \; e_l'(G-\hat{G}))' 
\\& =(e_1'(G-G^1+G^1-\hat{G}^1+\hat{G^1}-\hat{G})
\\& \hspace{7em} \ldots \; e_l'(G-G^l+G^l-\hat{G}^l+\hat{G^l}-\hat{G}))'
\\& =(e_1'(G-G^1+\hat{G^1}-\hat{G}) \; \ldots \; e_l'(G-G^l+\hat{G^l}-\hat{G}))'.
\end{align*}
Let $x_k=G-G^k+\hat{G^k}-\hat{G}$. We then have $\Vert E \Delta G \Vert_{\infty} \leq \max\limits_k \Vert x_k \Vert_{\infty}$ as $e^k$ has only elements between $0$ and $1$. Therefore, by the triangle inequality,
$$ \Vert E \Delta G \Vert_{\infty} \leq 2\delta.$$
It immediately follows that $\Vert \Delta G \Vert_{\infty} \leq 2 \Vert E^{-1} \Vert_{\infty}  \cdot \delta.$
\end{proof}

\section{Writing Cournot constraints efficiently} 

\subsection{Casting the equilibrium constraints as linear constraints}\label{sec: cournot_eq}

When action profile $(q_1,\ldots,q_n)$ is chosen, player $i$ gets utility $u_i(q_i,q_{-i})=q_i P(q_1,\ldots,q_n) - c_i(q_i)$ where $q_{-i}$ denotes the production levels of all players but $i$. A pure action profile $q^*=(q_1^*,\ldots,q_n^*)$ is a Nash Equilibrium if and only if for all players $i$, $q_i^*$ maximizes $u_i(q_i,q_{-i}^*)$; as $P$ is concave and $c_i$ is convex, $u_i$ is convex in $q_i$ and the equilibrium condition is equivalent to the first order condition
\begin{equation}\label{eq: FOC}
q_i \frac{\partial P}{\partial q_i}(q_1,\ldots,q_n) + P(q_1,\ldots,q_n) =c_i'(q_i), \; \forall i
\end{equation}
Then, combining Equations~\ref{eq: FOC} and~\ref{eq: poly_cost}, the equilibrium constraints become
\begin{equation}\label{eq: cournot_eq_constraint}
q_i \frac{\partial P}{\partial q_i}(q_1,\ldots,q_n) + P(q_1,\ldots,q_n) =\sum\limits_{k=1}^d k a_i(k) q_i^{k-1}
\end{equation}
which are linear and tractable in the variables $(a_i(0),a_i(1),\ldots,a_i(d)))$, as long as $P(q_1,\ldots,q_n),\frac{\partial P}{\partial q_i}(q_1,\ldots,q_n)$ can be efficiently computed given observations $q_1,\ldots,q_n$. Such equilibrium constraints can be incorporated into our framework. 

\subsection{Casting the convexity constraints as SDP constraints}\label{sec: cournot_convex}

We need to be able to deal with``$c_i$ is convex polynomial of degree $d$'' constraints for all $i$, in a computationally efficient manner. This constraint can be rewritten as ``$c_i''$ is a non-negative polynomial of degree $d$.'' Fortunately, this is a classic constraint in the realm of convex optimization, and can be dealt with in the following ways:
\begin{itemize}
\item If $d=1$, then $c_i(q_i)=a_i(1) q_i + a_i(0)$ ($c_i''=0$) is always convex. In this case, no constraint need be added.
\item If $d \geq 2$, then we need to ensure that $c''_i$ is non-negative in all points. It is known that a univariate polynomial is non-negative if and only if it can be written as a sum-of-squares; such constraints can efficiently be transformed into tractable semidefinite constraints. (For more details on the SDP formulation of sum-of-squares constraints, see~\cite{parrilo2004sum}.)
\end{itemize}